\documentclass{lmcs}
\pdfoutput=1

\usepackage{lastpage}
\lmcsdoi{16}{2}{3}
\lmcsheading{}{\pageref{LastPage}}{}{}%
{Oct.~08,~2019}{May~13,~2020}{}

\usepackage{eqparbox}
\usepackage{amsmath}
\usepackage{latexsym, amsfonts, stmaryrd, amssymb, graphicx, wrapfig, color}
\usepackage{tikz}
\usetikzlibrary{arrows,arrows.meta,shapes,snakes,automata,backgrounds,petri,
decorations.pathmorphing,positioning,fit}

\theoremstyle{plain}
\newtheorem{definition}[thm]{Definition}
\newtheorem{theorem}[thm]{Theorem}
\newtheorem{proposition}[thm]{Proposition}
\newtheorem{example}[thm]{Example}
\newtheorem{corollary}[thm]{Corollary}
\newtheorem{remark}[thm]{Remark}

\newcommand{\sharpbin}{\ensuremath{\mathbin{\#}}}
\newcommand{\sharphatbin}{\ensuremath{\mathbin{\hat{\#}}}}
\newcommand{\CF}[1]{\ensuremath{\mathsf{CF}(#1)}}
\newcommand{\trans}[1]{\ensuremath{\,[\/{#1}\/\rangle}\,}
\newcommand{\hist}[1]{\ensuremath{\lfloor #1 \rfloor}}
\newcommand{\Pow}[1]{\ensuremath{\mathbf{2}^{#1}}}
\newcommand{\Powfin}[1]{\ensuremath{\mathbf{2}_\mathit{fin}^{#1}}}
\newcommand{\Powone}[1]{\ensuremath{\mathbf{2}_1^{#1}}}
\newcommand{\nat}{\ensuremath{\mathbb{N}}}

\newcommand{\pes}{\textsc{pes}}
\newcommand{\fes}{\textsc{fes}}
\newcommand{\ies}{\textsc{ies}}
\newcommand{\rces}{\textsc{rces}}

\newcommand{\ses}{\textsc{ses}}
\newcommand{\ges}{\textsc{ges}}
\newcommand{\dces}{\textsc{dces}}

\newcommand{\Ges}{\textsc{cdes}}
\newcommand{\cd}{\textsc{cd}}
\newcommand{\ctx}{\textsc{Cxt}}
\newcommand{\gesrel}{\ensuremath{\gg}}
\newcommand{\enab}[1]{\ensuremath{[\/{#1}\/\rangle}}
\newcommand{\setenum}[1]{\{#1\}}
\newcommand{\setcomp}[2]{\{{#1} \mid {#2}\}}
\newcommand{\card}[1]{\ensuremath{|#1|}}
\newcommand{\shrinkt}{\ensuremath{\lhd}}
\newcommand{\growt}{\ensuremath{\blacktriangleright}}
\newcommand{\shrink}[3]{#3\shrinkt[#1\rightarrow#2]}
\newcommand{\grow}[3]{#1\ensuremath{\blacktriangleright}[#2\rightarrow#3]}
\DeclareMathOperator{\ic}{ic} 
\DeclareMathOperator{\ac}{ac} 
\DeclareMathOperator{\dc}{dc} 
\DeclareMathOperator{\Ctr}{\mathcal{C}}
\DeclareMathOperator{\Rtc}{\mathcal{I}}
\newcommand{\pmv}[1]{\ensuremath{\mathsf{#1}}}
\newcommand{\Conf}[2]{\ensuremath{\mathsf{Conf}_{#2}(#1)}}
\newcommand{\shrset}[1]{\ensuremath{\mathsf{ShrMod}(#1)}}
\newcommand{\groset}[1]{\ensuremath{\mathsf{GroMod}(#1)}}
\newcommand{\drop}[2]{\ensuremath{\mathsf{Drop}({#1},{#2})}}
\newcommand{\agg}[2]{\ensuremath{\mathsf{Add}({#1},{#2})}}
\newcommand{\grafo}[2]{\ensuremath{\mathcal{G}_{#1}({#2})}}
\newcommand{\toges}[2]{\ensuremath{\mathcal{F}_{#1}({#2})}}

\newcommand{\inh}{\ensuremath{\vdash\hspace*{-3mm}\mbox{\raisebox{1pt}{$\multimap$}}}}
\newcommand{\De}[3]{\ensuremath{\inh({#1},{#2},{#3})}}
\newcommand{\event}[1]{\ensuremath{\pmv{#1}}}
\newcommand{\len}[1]{\ensuremath{\mathit{len}(#1)}}
\newcommand{\tuple}[1]{\ensuremath{\langle #1\rangle}}
\newcommand{\stati}{\ensuremath{\mathsf{S}}}
\newcommand{\instate}{\ensuremath{\mathit{s}_0}}
\newcommand{\ea}{\ensuremath{\mathsf{ea}}}
\newcommand{\ragg}[1]{\ensuremath{\mathsf{reach}(#1)}}
\newcommand{\toset}[1]{\ensuremath{\overline{#1}}}
\newcommand{\fl}[1]{\ensuremath{\mathsf{fl}(#1)}}
\newcommand{\maxfl}[1]{\ensuremath{\mathsf{maxfl}(#1)}}

\begin{document}

\title{Representing Dependencies in Event Structures}

\author[G.M. Pinna]{G. Michele Pinna} 
\address{Dipartimento di Matematica e Informatica, Universit\`a di Cagliari, Cagliari, Italy}  
\email{gmpinna@unica.it}

  \begin{abstract}
     Event structures where the causality may explicitly change during a computation have  
     recently gained the stage.
     In this kind of event structures the changes in the set of the causes of an event are triggered 
     by modifiers that may add or remove dependencies, thus making the happening of an event contextual. 
     Still the focus is always on the dependencies of the event.
     In this paper we promote the idea that the \emph{context} determined by the modifiers plays a 
     major role, and the context itself determines not only the causes but also what causality 
     should be. Modifiers are then used to understand when an event (or a set of events) can be added 
     to a 
     configuration, together with a set of events modeling dependencies, which will play a less important 
     role.  We show that most of the notions of Event Structure presented in literature can be 
     translated into this new kind of event structure, preserving the main notion, namely the one of 
     configuration.  
  \end{abstract}

\maketitle



\section{Introduction}\label{sec:intro}
The notion of causality is an intriguing one. In the sequential case, 
the intuition behind it is almost trivial: 
if the activity $\pmv{e}$ depends on the activity $\pmv{e}'$, then to 
happen the activity $\pmv{e}$ needs that $\pmv{e}'$ has already happened. This is easily
represented in Petri nets (\cite{Rei:PNI}), the transition  $\pmv{e}'$ \emph{produces} a token
that is \emph{consumed} by the transition $\pmv{e}$ (the net $N'$ in Figure~\ref{fig:uno}). 
The dependency is testified by the observation
that the activity  $\pmv{e}'$ always precedes the activity  $\pmv{e}$.
However this intuition does not reflect other possibilities.
If we abandon the sequential case and move toward possibly loosely cooperating system the notion
of causality become involved.
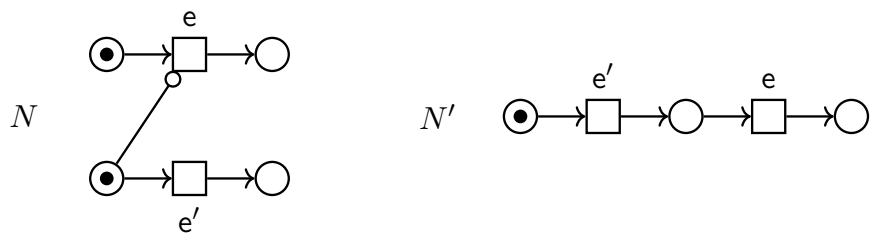
\begin{figure}[!h]
\centerline{\scalebox{1.1}{\begin{tikzpicture}
\tikzstyle{inhibitor}=[o-,thick]
\tikzstyle{pre}=[<-,thick]
\tikzstyle{post}=[->,thick]
\tikzstyle{transition}=[rectangle, draw=black,thick,minimum size=4mm]
\tikzstyle{place}=[circle, draw=black,thick,minimum size=4mm]
\node (n) at (-1,.75) {$N$};
\node (n1) at (4,.75) {$N'$};
\node[place,tokens=1] (p1) at (0,0) {};
\node[place,tokens=1] (p3) at (0,1.5) {};
\node[place] (p2) at (2,0) {};
\node[place] (p4) at (2,1.5) {};
\node[transition] (t1) at (1,0) [label=below:$\pmv{e}'$] {}
edge[pre] (p1)
edge[post] (p2);
\node[transition] (t2) at (1,1.5) [label=above:$\pmv{e}$] {}
edge[pre] (p3)
edge[inhibitor] (p1)
edge[post] (p4);
\node[place,tokens=1] (q1) at (5,.75) {};
\node[place] (q2) at (7,.75) {};
\node[place] (q3) at (9,.75) {};
\node[transition] (t3) at (6,.75) [label=above:$\pmv{e}'$] {}
edge[pre] (q1)
edge[post] (q2);
\node[transition] (t4) at (8,.75) [label=above:$\pmv{e}$] {}
edge[pre] (q2)
edge[post] (q3);
\end{tikzpicture}}}
\caption{Two Petri nets fostering different interpretations of the dependencies among events}\label{fig:uno}
\end{figure}
Consider the case of a Petri net with inhibitor arcs (\cite{JK:SIN}) where the precondition
of the transition $\pmv{e}'$ inhibits the transition $\pmv{e}$ (the net $N$ in Figure~\ref{fig:uno}). 
The latter to happens needs
that the transition $\pmv{e}'$ happens first, and the \emph{observation} testifies that the
activity $\pmv{e}$ needs that $\pmv{e}'$ has already happened, though resources are not exchanged 
between $\pmv{e}'$ and $\pmv{e}$, as it is done in the net $N'$ in Figure~\ref{fig:uno}.
In both cases the observation that the event $\pmv{e}'$ must happen first leads to state that
$\pmv{e}'$ precedes $\pmv{e}$ and this can be well represented with a partial order relation among
events.

This quite simple discussion suggests that the notion of causality may have many facets. 
In fact, if the dependencies are modeled
just with a well founded partial order, inhibitor arcs can be used to model 
these dependencies, but the notion of partial order does not capture precisely the
subtleties that are connected to the notion of causality.

To represent the semantics of concurrent and distributed systems 
the notion of \emph{event structure} plays a prominent role. 
Event structures have been introduced in \cite{NPW:PNES} and \cite{Win:ES} and since then 
have been considered as a cornerstone. 
The idea is simple: the activities of a system are the \emph{events} and their relationships are specified
somehow, \emph{e.g.} with a partial order modeling the \emph{enabling} and a 
predicate expressing when activities are
\emph{coherent} or not. 
Starting from this idea many authors have faced the problem of adapting this notion to many different
situations which have as a target the attempt to represent faithfully various situations.
This has triggered many diverse approaches. 
In \cite{Gun:GLCES} and \cite{Gu:CA} \emph{causal automata} are
discussed, with the idea that the conditions under which an event may happen are specified by a suitable 
logic formula, 
in \cite{GP87:POMCCF} and \cite{G88:MCPONTS} it is argued that a partial order may be not enough
or may be, in some situation, a too rigid notion, and this idea is used also in \cite{PP:NE} and
\cite{PP:NEPC} where the notion of \emph{event automata} is introduced, 
and it is used also in \cite{Pi:FI06} where an enabling/disabling relation for event 
automata is discussed.
Looking at the enabling relation, both \emph{bundle event structures} (\cite{Lan92:BES}) and 
\emph{dual event structures} (\cite{LBK:dual}) provide
a notion of enabling capturing \emph{or}-causality (the former exclusive \emph{or}-causality and 
the latter
non exclusive \emph{or}-causality).
\emph{Asymmetric event structures} (\cite{BCM:CNAED})
introduces a weaker notion of causality which models contextual arcs in Petri nets, or
in the case \emph{circular event structures} (\cite{BCPZ:CircCausES}) the enabling notion is tailored to
model also \emph{circular dependencies}. In \emph{flow event structures} (\cite{Bou:FES}) the partial 
order is required to hold only in configurations.
An approach 
where dependencies may change according to the presence of \emph{inhibitions} 
is proposed in \cite{BBCP:FCSPNRI} and \cite{BBCP:rivista}, where \emph{inhibitor event structures}
are studied.  
This short and incomplete discussion (the event structures spectrum is rather broad) 
should point out the variety of approaches present in literature.
It should be also observed that the majority of the approaches model causality with a relation
that can be reduced to a partial order, hence causality is represented stating what are the
events that should have happened before. 

In this paper we introduce yet another notion of event structure. 
Triggered by recent works on \emph{adding} or 
\emph{subtracting} dependencies among events based on the fact that apparently
unrelated events have happened (\cite{AKPN:DC,AKPN:lmcs18}), we argue that rather than focussing on how
to model these enrichment or/and impoverishment, it is much more natural to focus on the context
on which an event takes place. 
In fact it is the context that can determine the proper dependencies that are applicable at the state
where the event should take place and the context holds, 
and the context can also be used as well to forbid that the event is added to the state.
This new relation resembles the one used in inhibitor event structures, but 
it differs in the way the contexts are determined. In the case of inhibitor event structures 
the presence of a certain event (the inhibiting
context) was used to require that another one was present as well (representing the trigger able
to remove the inhibition). 
Here the flavour is different as it is more prescriptive: it is required that exactly a set of events 
is present and if this happens then also another one should be present as well.
It should be stressed that triggers and contexts may exchange their role.
Consider again the two nets depicted in Figure~\ref{fig:uno}, 
we may have that in both cases the trigger is 
determined by the happening of the event $\pmv{e}'$ and the context is the empty set, but we can 
consider as context the event $\pmv{e}'$  and the trigger as the empty set.
This simple relation, which we will call \emph{context-dependency} relation, 
suffices to cover the aforementioned notions.
It is worth observing that determining the context and the triggers associated to it is quite 
similar to
trying to understand the dependencies. Consider the net $N''$ in Figure~\ref{fig:due}.
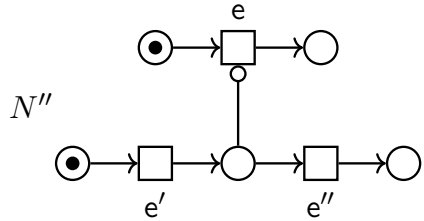
\begin{figure}[h]
\centerline{\scalebox{1.1}{\begin{tikzpicture}
\tikzstyle{inhibitor}=[o-,thick]
\tikzstyle{pre}=[<-,thick]
\tikzstyle{post}=[->,thick]
\tikzstyle{transition}=[rectangle, draw=black,thick,minimum size=4mm]
\tikzstyle{place}=[circle, draw=black,thick,minimum size=4mm]
\node (n) at (-.5,.7) {$N''$};
\node[place,tokens=1] (p1) at (0,0) {};
\node[place,tokens=1] (p3) at (1,1.4) {};
\node[place] (p2) at (2,0) {};
\node[place] (p4) at (3,1.4) {};
\node[place] (p5) at (4,0) {};
\node[transition] (t1) at (1,0) [label=below:$\pmv{e}'$] {}
edge[pre] (p1)
edge[post] (p2);
\node[transition] (t2) at (2,1.4) [label=above:$\pmv{e}$] {}
edge[pre] (p3)
edge[inhibitor] (p2)
edge[post] (p4);
\node[transition] (t3) at (3,0) [label=below:$\pmv{e}''$] {}
edge[pre] (p2)
edge[post] (p5);
\end{tikzpicture}}}
\caption{A net where the inhibition can be removed, and the dependency among events is 
contextual}\label{fig:due}
\end{figure}
Here $\pmv{e}$ may be added either to the empty set or to a set containing both 
$\pmv{e}'$ and $\pmv{e}''$. The context containing $\pmv{e}'$ only leads to require that the 
event $\pmv{e}''$ is present (in the spirit of the relation for inhibitor event structures), making
$\pmv{e}$ dependent on $\pmv{e}''$. However we could also have chosen to focus on contexts 
only and in this case
the context containing just $\pmv{e}'$ is ruled out among the contexts in which, 
together with some others
dependencies, $\pmv{e}$ may take place, and in this case the two contexts are $\emptyset$ and 
$\setenum{\pmv{e}', \pmv{e}''}$.
As hinted above, it will turn out that the context plays a more relevant 
role with respect to the dependencies,
as the context can be seen positively 
(it specifies under which conditions an event can be added, together
with the dependency) or negatively (it specifies  under which conditions an 
event cannot be added, and in this case
the event cannot be added simply stipulating that it depends on itself).

In this paper we will focus on event structures where the change of state is always triggered by 
the happening of a
single event, hence we will not consider steps (\emph{i.e.} non empty and finite subsets of 
events). 
However the generalization to steps is straightforward.
We will also assume that in the event structures considered, each configuration can be reached by the 
initial one.

This paper is an extended and revised version of \cite{Pi:coord}. We have added some examples  and
further discussed the notion of context-dependent event structure. We also have made
more precise its relationship with various kind of event structures, among them flow event structures
and inhibitor event structures, and 
we have better detailed the characteristics of the considered event automata.

\subsubsection*{Organization of the paper and contributions:}\ 
In the next section we will introduce and discuss the new brand of event structure, 
\emph{context-dependent} event structure,
which is the main contribution of this paper and we will provide various examples
to show how the new relation is used to model different situations. 
In Section~\ref{sec:es} 
we review and briefly analyze
some notions of event structures presented in literature, 
namely \emph{prime} event structure (\cite{Win:ES}), 
\emph{flow} event structure (\cite{Bou:FES}),
\emph{relaxed prime} event structure and \emph{dynamic 
causality} event structure (\cite{AKPN:DC}), \emph{inhibitor} event structure (\cite{BBCP:rivista}) and
event structure for \emph{resolvable conflicts} (\cite{GP:ESRC}), and,
in section~\ref{sec:relating}, we show that the each event structure
presented in section~\ref{sec:es} can be translated into this new kind of event structure. 
This will be obtained coding each event structure into an \emph{event automaton} (\cite{PP:NEPC})
and then showing how to associate a context-dependent event structure to each event automaton.
We will also briefly discuss the expressivity of inhibitor event structure with respect to
the event structure where dependencies may grow, adding a tiny to information on how 
the various kind of event structures are related. 
We will end the paper  with some conclusions and we will give some hints for further developments.

\subsubsection*{Notation:}\
Let $A$ be a set, with let $\rho$ we denote a sequence of elements 
belonging to $A$,
and with $\epsilon$ we denote the empty sequence. With $\overline{\rho}$ we denote the
set of elements of $A$ appearing in $\rho$. Thus  
$\overline{\rho} = \emptyset$ if $\rho = \epsilon$ and 
$\overline{\pmv{a}\rho'} = \setenum{\pmv{a}}\cup \overline{\rho'}$ if $\rho = \pmv{a}\rho'$.
Given a sequence $\rho = \pmv{a}_1\cdots\pmv{a}_{n}\cdots$ with $\len{\rho}$ we denote its length.
With $\rho_0$ we sometimes denote sequence $\epsilon$ and, if $\len{\rho}\geq 1$,  
for each $1 \leq i \leq \len{\rho}$ with $\rho_i$ we denote the sequence $\pmv{a}_1\cdots\pmv{a}_{i}$.
A sequence $\rho$ can be infinite (and then it can be seen as a mapping from $\nat$ to $A$). 
Let $A$ be a set, with $\Pow{A}$ we denote the subsets of $A$ and with $\Powfin{A}$ the finite 
subsets of $A$.


\section{Context-Dependent Event Structure}\label{sec:es-mie}
We introduce yet another notion of event structure, which is the main contribution 
of the paper.

We recall what an \emph{event} is. 
An event is an \emph{atomic} and \emph{individual} action which is able to change the 
state of a system.
Though this may appear totally obvious, we believe it is better to stress these characteristics, namely
atomicity, individuality and the ability to change the state of the system, which means that we
can observe events, and their happening as a whole. 

Event structures are intended to model \emph{concurrent} and \emph{distributed} systems where the
activities are represented by events by
defining relationships among these events such as \emph{causality} and
\emph{conflict}, and establishing the conditions on which a certain event can be added to
a state. 
The \emph{state} of a system modeled by an event structure is a \emph{subset} of events (those happened
so far, representing the accomplished activities), and this set of events
is often called \emph{configuration}. 
States can be enriched by adding other information beside the one
represented by the events that have determined the state, either adding information on the relationship
among the various events in the state, \emph{e.g.} adding dependencies among them (the state is then 
a partial order, \cite{Ren:concur92} or \cite{AKPN:lmcs18}) 
or adding suitable information to the whole state.

We pursue the idea that the happening of an event depends on a set of modifiers (the \emph{context}) 
and on a set of \emph{real} dependencies, which are activated by the set of modifiers. 

We recall that in this paper we will consider only \emph{unlabelled} event structures.
To simplify the presentation we retain a binary conflict relation. 
We first define what a \emph{conflict free} subset of events is.
\begin{definition}\label{de:conf-free}
Let $\event{E}$  be a set of events and let $\# \subseteq \event{E}\times\event{E}$ be an irreflexive 
and symmetric relation, called \emph{conflict} relation. 
Let $X\subseteq \event{E}$ be a subset of events, we say that $X$ is \emph{conflict free}, denoted
with $\CF{X}$, iff $\forall \pmv{e}, \pmv{e}'\in X$ it holds that 
$\neg(\pmv{e} \mathbin{\#} \pmv{e}')$.
\end{definition}
We introduce the notion of \emph{context-dependent} event structure.

\begin{definition}\label{de:mia-es}
 A \emph{context-dependent event structure} ({\Ges}) is a triple 
 $\mathit{E} = (\event{E}, \#, \gesrel)$ where 
 \begin{itemize}
  \item 
  $\event{E}$ is a set of \emph{events},
  \item 
  $\mathord{\#} \subseteq \event{E}\times \event{E}$ is an irreflexive and symmetric relation, called 
        \emph{conflict relation}, and
  \item 
  $\gesrel\ \subseteq \Pow{\pmv{A}}\times \event{E}$,  
        where $\pmv{A} \subseteq \Powfin{\event{E}}\times\Powfin{\event{E}}$,   
        is a relation, called
        the \emph{context-dependency} relation (\cd-relation), which is such that for each 
        $\pmv{Z}\gesrel \pmv{e}$ it holds that 
         \begin{itemize}
          \item 
                $\pmv{Z}\neq \emptyset$,
          \item 
                for each $(X,Y)\in \pmv{Z}$ it holds that $\CF{X}$ and
                $\CF{Y}$, 
                and
          \item for each $(X,Y), (X',Y')\in \pmv{Z}$ if $X = X'$ then $Y = Y'$.
         \end{itemize}
 \end{itemize}
\end{definition}
Each element of the \cd-relation $\gesrel$ will be called \emph{entry}, and for each
entry $\pmv{Z} \gesrel \pmv{e}$, we will call $(X,Y) \in \pmv{Z}$ the \emph{element} of the
entry.
Finally, for each element $(X,Y)$ of each entry, $X$ is the set of \emph{modifiers} and
$Y$ is the set of \emph{dependencies}.
The \cd-relation models, for each event, which are the possible contexts in which the event may happen 
(the modifiers of each element) and for each context which are the events that have to be occurred 
(the dependencies).
We stipulate that dependencies and modifiers are formed by non conflicting events, 
though this is not strictly needed, as
the relation can model also conflicts. 
Conflicts can be modeled making the event dependent on itself in the presence of a suitable
set of modifiers, as we will point out later.
We also require that, for each entry of the \cd-relation, two elements of the entry cannot have the
same modifiers.

The first step toward showing how this relation is used is formalized in 
the notion of enabling of an event. 
We have to determine, for each $\pmv{Z} \gesrel \pmv{e}$, which of the set of modifiers
$X_i$ should be considered, \emph{i.e.} which element of the entry should be used. 
To do so we define the \emph{context} associated to each 
entry of the \cd-relation.
\begin{definition}\label{de:mia-context-Z}
Let  $\mathit{E} = (\event{E}, \#, \gesrel)$  be a \Ges.
Let $\pmv{Z} \gesrel \pmv{e}$ be an entry of $\gesrel$, with 
$\pmv{Z} = \setenum{(X_1,Y_1), (X_2,Y_2), \dots}$, 
with $\ctx(\pmv{Z}\gesrel\pmv{e})$ we
denote the set of events $\bigcup_{i=1}^{|Z|} X_i$.
\end{definition}
The context associated to the entry $\pmv{Z}\gesrel\pmv{e}$ is the union of all the set of
modifiers for the given entry. 

In the definition of \emph{enabling} of an event $\pmv{e}$ at a subset of events $C$ it will be clear
what is the role of a context. 
\begin{definition}\label{de:mia-es-enabling}
 Let $\mathit{E} = (\event{E}, \#, \gesrel)$ be a \Ges\ and let $C\subseteq \event{E}$ 
 be a subset of events. 
 The event $\pmv{e}\not\in C$ is \emph{enabled} at $C$, denoted with $C\enab{\pmv{e}}$, 
 if for each entry $\pmv{Z} \gesrel \pmv{e}$, with 
 $\pmv{Z} = \setenum{(X_1,Y_1), (X_n,Y_n), \dots}$, there is an element 
 $(X_i,Y_i)\in \pmv{Z}$ such that
 \begin{itemize}
   \item $\ctx(\pmv{Z} \gesrel \pmv{e})\cap C = X_i$, and 
   \item $Y_i\subseteq C$.
 \end{itemize}  
\end{definition}
An event $\pmv{e}$, not present in a subset of event, is enabled at this subset whenever, for each entry
of the \cd-relation, there exists an element such that the set of modifiers is equal to the
intersection of the context of this entry and the subset of events.
Observe that requiring the non emptiness of the set $\pmv{Z}$ in $\pmv{Z} \gesrel \pmv{e}$ 
guarantees that 
an event $\pmv{e}$ may be enabled at some subset of events. 
The requirement that two different elements of an entry must have different sets of modifiers
is used to rule out the possibility that the same event is enabled by one element in more
than one different ways.
We do not make any further assumption on the elements of an entry: the sets
of modifiers must be different but they can be in any other relation among sets, in 
particular a set of modifiers can be contained into another. Consider the
events involved in net depicted in Figure~\ref{fig:due}, we have that the event (transition) $\pmv{e}$
can happen at the beginning (hence with $\emptyset$ as set of modifiers) or
when $\pmv{e}'$ and $\pmv{e}''$ have happened, and in any way this will be modeled, the set of 
modifiers would be
non empty (either containing just $\pmv{e}'$ or also $\pmv{e}''$), 
containing obviously the previous set of modifiers. 
As we have already pointed out, the \cd-relation could be used to express conflicts: 
$\pmv{e}\sharpbin\pmv{e'}$ could be modeled by adding
$\setenum{(\emptyset,\emptyset),(\setenum{\pmv{e}},\setenum{\pmv{e}'})} \gesrel \pmv{e}'$ and 
$\setenum{(\emptyset,\emptyset),(\setenum{\pmv{e}'},\setenum{\pmv{e}})} \gesrel \pmv{e}$ to the 
$\gesrel$ relation, 
and the presence of just one of them would model the
asymmetric conflict. The conflicts modeled in this way are \emph{persistent}, namely they 
cannot be changed.

We define now what is the state of system modeled by a \Ges. 
\begin{definition}
\label{de:mia-es-event-traces}
Let $\mathit{E} = (\event{E}, \#, \gesrel)$ be a \Ges. 
Let $C$ be a subset of $\event{E}$. 
We say that $C$ is a \emph{configuration} of the \Ges\ $\mathit{E}$ 
iff there exists a 
sequence of distinct events $\rho = \pmv{e}_1\cdots\pmv{e}_n\cdots$ over $\event{E}$ 
such that 
\begin{itemize}
  \item 
  $\overline{\rho} = C$,
  \item 
  $\CF{\overline{\rho}}$, and 
  \item 
  $\forall 1 \leq i \leq \len{\rho}.\ \overline{\rho}_{i-1} \enab{\pmv{e}_i}$.
\end{itemize}
With $\Conf{\mathit{E}}{\Ges}$ we denote 
the set of configurations of the \Ges\ $\mathit{E}$.
\end{definition}
This definition is standard, 
and it has an operational flavour as an ordering in executing the
events should be found such that at each step one of the events can be added as it is
enabled. 
The sequence of distinct events $\rho = \pmv{e}_1\cdots\pmv{e}_n\cdots$ will be sometimes called
\emph{event trace}.
In a \Ges\ the empty set is always a configuration, as it is conflict free and the empty 
sequence of events gives the proper events trace.

The following definition introduces a relation $\mapsto_{\Ges}$ among configurations stating when
a configuration can be \emph{reached} from another one by adding an event. 
\begin{definition}\label{de:mia-set-configuration}
Let $\mathit{E} = (\event{E}, \#, \gesrel)$ be a \Ges. 
Given two configurations $C, C'\in \Conf{\mathit{E}}{\Ges}$ such that 
$C'\setminus C = \setenum{\pmv{e}}$, we stipulate that 
$C \mapsto_{\Ges} C'$ iff $C\enab{\pmv{e}}$.
\end{definition}
We illustrate this new kind of event structure with some examples, that should clarify how the 
\cd-relation is used and give some hints on its \emph{expressivity}. 

\begin{figure}[ht]
\centerline{\begin{tabular}{ccc}
\begin{tikzpicture}
[bend angle=45, scale=0.9, 
  read/.style={-,shorten
    <=0pt,thick},
  pre/.style={<-,shorten
    <=0pt,>=stealth,>={Latex[width=1mm,length=1mm]},thick}, post/.style={->,shorten
    >=0,>=stealth,>={Latex[width=1mm,length=1mm]},thick},place/.style={circle, draw=black,
    thick,minimum size=4mm}, transition/.style={rectangle, draw=black!0,
    thick,minimum size=4mm}, invplace/.style={circle,
    draw=black!0,thick}]
\node[transition] (spe) at (0,1.5) {};  
\node[transition] (in) at (0,0) {$\emptyset$};
\node[transition] (a) at (1.8,1) {$\setenum{\pmv{a}}$}
edge[pre] (in);
\node[transition] (b) at (1.8,-1) {$\setenum{\pmv{b}}$}
edge[pre] (in);
\node[transition] (c) at (1.8,0) {$\setenum{\pmv{c}}$}
edge[pre] (in);
\node[transition] (bc) at (4,-1) {$\setenum{\pmv{b},\pmv{c}}$}
edge[pre] (c)
edge[pre] (b);
\node[transition] (ac) at (4,1) {$\setenum{\pmv{a},\pmv{c}}$}
edge[pre] (a)
edge[pre] (c);
\node[transition] (abc) at (7,0) {$\setenum{\pmv{a},\pmv{b},\pmv{c}}$}
edge[pre] (bc)
edge[pre] (ac);
\end{tikzpicture} & & \begin{tikzpicture}
[bend angle=45, scale=0.9, 
  read/.style={-,shorten
    <=0pt,thick},
  pre/.style={<-,shorten
    <=0pt,>=stealth,>={Latex[width=1mm,length=1mm]},thick}, post/.style={->,shorten
    >=0,>=stealth,>={Latex[width=1mm,length=1mm]},thick},place/.style={circle, draw=black,
    thick,minimum size=4mm}, transition/.style={rectangle, draw=black!0,
    thick,minimum size=4mm}, invplace/.style={circle,
    draw=black!0,thick}]
\node[transition] (spe) at (0,1.5) {};  
\node[transition] (in) at (0,0) {$\emptyset$};
\node[transition] (a) at (1.8,1) {$\setenum{\pmv{a}}$}
edge[pre] (in);
\node[transition] (b) at (1.8,-1) {$\setenum{\pmv{b}}$}
edge[pre] (in);
\node[transition] (ab) at (4,0) {$\setenum{\pmv{a},\pmv{b}}$}
edge[pre] (a)
edge[pre] (b);
\node[transition] (ac) at (4,1) {$\setenum{\pmv{a},\pmv{c}}$}
edge[pre] (a);
\node[transition] (bc) at (4,-1) {$\setenum{\pmv{b},\pmv{c}}$}
edge[pre] (b);
\node[transition] (abc) at (7,0) {$\setenum{\pmv{a},\pmv{b},\pmv{c}}$}
edge[pre] (ab)
edge[pre] (ac)
edge[pre] (bc);
\end{tikzpicture} \\ [3mm]
 (a) & & (b) \\ [6mm]
\begin{tikzpicture}
[bend angle=45, scale=0.9, 
  read/.style={-,shorten
    <=0pt,thick},
  pre/.style={<-,shorten
    <=0pt,>=stealth,>={Latex[width=1mm,length=1mm]},thick}, post/.style={->,shorten
    >=0,>=stealth,>={Latex[width=1mm,length=1mm]},thick},place/.style={circle, draw=black,
    thick,minimum size=4mm}, transition/.style={rectangle, draw=black!0,
    thick,minimum size=4mm}, invplace/.style={circle,
    draw=black!0,thick}]
\node[transition] (spe) at (0,1.5) {};  
\node[transition] (in) at (0,0) {$\emptyset$};
\node[transition] (a) at (1.8,1) {$\setenum{\pmv{a}}$}
edge[pre] (in);
\node[transition] (c) at (1.8,-1) {$\setenum{\pmv{b}}$}
edge[pre] (in);
\node[transition] (ac) at (4,1) {$\setenum{\pmv{a},\pmv{c}}$}
edge[pre] (a);
\node[transition] (bc) at (4,-1) {$\setenum{\pmv{b},\pmv{c}}$}
edge[pre] (c);
\end{tikzpicture} & & \begin{tikzpicture}
[bend angle=45, scale=0.9, 
  read/.style={-,shorten
    <=0pt,thick},
  pre/.style={<-,shorten
    <=0pt,>=stealth,>={Latex[width=1mm,length=1mm]},thick}, post/.style={->,shorten
    >=0,>=stealth,>={Latex[width=1mm,length=1mm]},thick},place/.style={circle, draw=black,
    thick,minimum size=4mm}, transition/.style={rectangle, draw=black!0,
    thick,minimum size=4mm}, invplace/.style={circle,
    draw=black!0,thick}]
\node[transition] (spe) at (0,1.5) {};  
\node[transition] (in) at (0,0) {$\emptyset$};
\node[transition] (a) at (1.8,1) {$\setenum{\pmv{a}}$}
edge[pre] (in);
\node[transition] (c) at (1.8,-1) {$\setenum{\pmv{b}}$}
edge[pre] (in);
\node[transition] (ab) at (4,0) {$\setenum{\pmv{a},\pmv{b}}$}
edge[pre] (a);
\end{tikzpicture} \\ [3mm]
 (c) & & (d) \\  [6mm]
\begin{tikzpicture}
[bend angle=45, scale=0.9, 
  read/.style={-,shorten
    <=0pt,thick},
  pre/.style={<-,shorten
    <=0pt,>=stealth,>={Latex[width=1mm,length=1mm]},thick}, post/.style={->,shorten
    >=0,>=stealth,>={Latex[width=1mm,length=1mm]},thick},place/.style={circle, draw=black,
    thick,minimum size=4mm}, transition/.style={rectangle, draw=black!0,
    thick,minimum size=4mm}, invplace/.style={circle,
    draw=black!0,thick}]
\node[transition] (spe) at (0,1.5) {};  
\node[transition] (in) at (0,0) {$\emptyset$};
\node[transition] (a) at (1.8,1) {$\setenum{\pmv{a}}$}
edge[pre] (in);
\node[transition] (c) at (1.8,-1) {$\setenum{\pmv{c}}$}
edge[pre] (in);
\node[transition] (b) at (1.8,0) {$\setenum{\pmv{b}}$}
edge[pre] (in);
\node[transition] (ac) at (4,0) {$\setenum{\pmv{a},\pmv{c}}$}
edge[pre] (a)
edge[pre] (c);
\node[transition] (ab) at (4,1) {$\setenum{\pmv{a},\pmv{b}}$}
edge[pre] (a)
edge[pre] (b);
\node[transition] (bc) at (4,-1) {$\setenum{\pmv{b},\pmv{c}}$}
edge[pre] (c);
\node[transition] (abc) at (7,0) {$\setenum{\pmv{a},\pmv{b},\pmv{c}}$}
edge[pre] (ab)
edge[pre] (ac)
edge[pre] (bc);
\end{tikzpicture} & & \begin{tikzpicture}
[bend angle=45, scale=0.9, 
  read/.style={-,shorten
    <=0pt,thick},
  pre/.style={<-,shorten
    <=0pt,>=stealth,>={Latex[width=1mm,length=1mm]},thick}, post/.style={->,shorten
    >=0,>=stealth,>={Latex[width=1mm,length=1mm]},thick},place/.style={circle, draw=black,
    thick,minimum size=4mm}, transition/.style={rectangle, draw=black!0,
    thick,minimum size=4mm}, invplace/.style={circle,
    draw=black!0,thick}]
\node[transition] (spe) at (0,1.5) {};  
\node[transition] (in) at (0,0) {$\emptyset$};
\node[transition] (a) at (1.8,1) {$\setenum{\pmv{a}}$}
edge[pre] (in);
\node[transition] (c) at (1.8,-1) {$\setenum{\pmv{c}}$}
edge[pre] (in);
\node[transition] (b) at (1.8,0) {$\setenum{\pmv{b}}$}
edge[pre] (in);
\node[transition] (ac) at (4,0) {$\setenum{\pmv{a},\pmv{c}}$}
edge[pre] (a)
edge[pre] (c);
\node[transition] (ab) at (4,1) {$\setenum{\pmv{a},\pmv{b}}$}
edge[pre] (a)
edge[pre] (b);
\node[transition] (bc) at (4,-1) {$\setenum{\pmv{b},\pmv{c}}$}
edge[pre] (c)
edge[pre] (b);
\end{tikzpicture} \\  [6mm]
(e) & & (f) \\ 
\end{tabular}}
\caption{The configurations of the \Ges\ of the
Examples~\ref{ex:new2res} (a), \ref{ex:shr-dep} (b),
\ref{ex:or-caus} (c), \ref{ex:as-conf} (d),  
\ref{ex:grow-dep} (e) and \ref{ex:new3conf} (f)}\label{fig:mia-examples}
\end{figure}
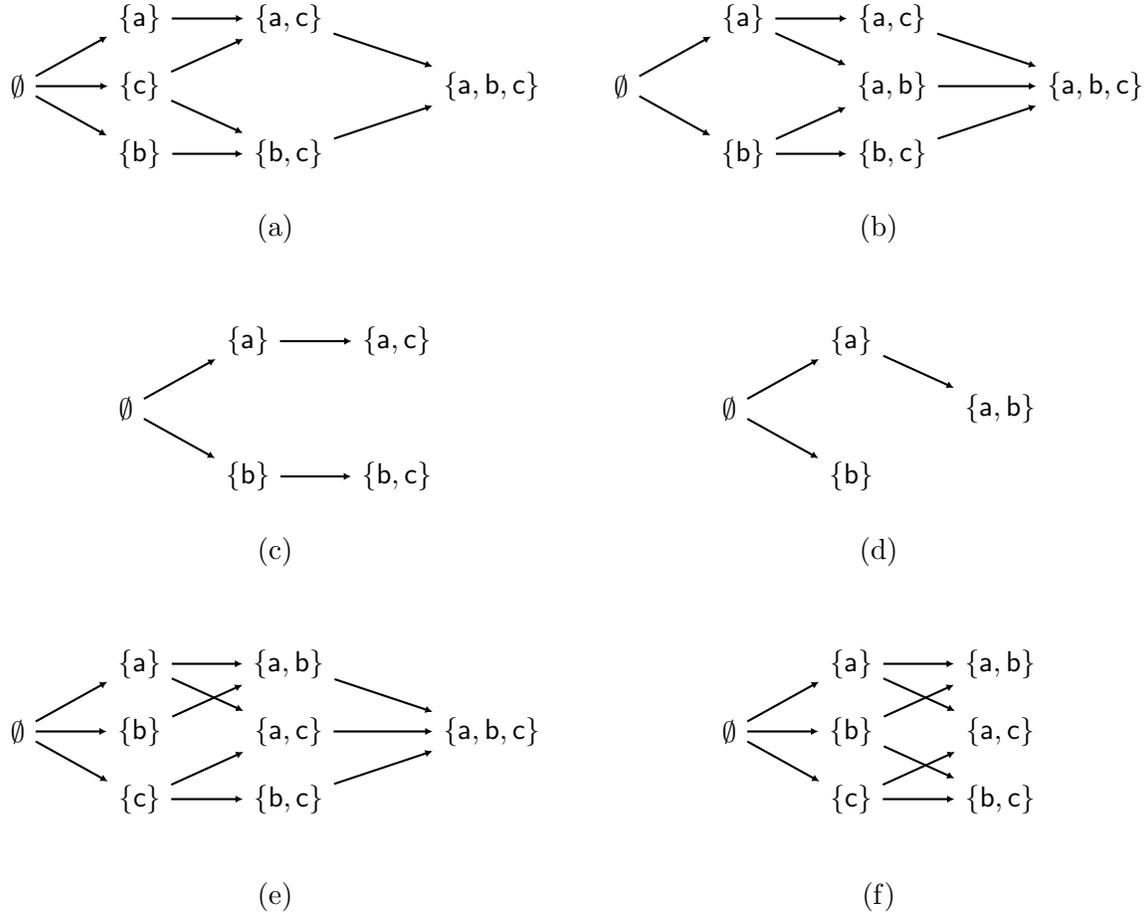

\begin{example}[Resolvable conflicts]\label{ex:new2res}
 Consider three events $\pmv{a}, \pmv{b}$ and $\pmv{c}$. All the events are 
 singularly enabled but $\pmv{a}$ and $\pmv{b}$ are in conflict unless $\pmv{c}$ has not happened
 (this are called \emph{resolvable} conflicts as the conflict between $\pmv{a}$ and $\pmv{b}$ is
 resolved by the execution of $\pmv{c}$).
 For the event $\pmv{a}$ we stipulate 
 $$\setenum{(\emptyset,\emptyset),(\setenum{\pmv{c}},\emptyset),(\setenum{\pmv{b}},\setenum{\pmv{c}})}
 \gesrel \pmv{a}$$
 that should be interpreted as follows: if the context is $\emptyset$ or $\setenum{\pmv{c}}$ then
 $\pmv{a}$ is enabled without any further condition (the $Y$ are the empty set), if the context
 is $\setenum{\pmv{b}}$ then also $\setenum{\pmv{c}}$ should be present.
 The set $\ctx(\setenum{(\emptyset,\emptyset),(\setenum{\pmv{c}},\emptyset),
 (\setenum{\pmv{b}},\setenum{\pmv{c}})}\gesrel \pmv{a})$ is $\setenum{\pmv{b},\pmv{c}}$. 
 Similarly, for the event
 $\pmv{b}$ we stipulate 
 $$\setenum{(\emptyset,\emptyset),(\setenum{\pmv{c}},\emptyset),(\setenum{\pmv{a}},\setenum{\pmv{c}})}\gesrel \pmv{b}$$ 
 which is justified as above, and such that 
 $\ctx(\setenum{(\emptyset,\emptyset),(\setenum{\pmv{c}},\emptyset),(\setenum{\pmv{a}},\setenum{\pmv{c}})}\gesrel \pmv{b})$ is $\setenum{\pmv{a},\pmv{c}}$.  
 Finally, for the event $\pmv{c}$, 
 we stipulate 
 $$\setenum{(\emptyset,\emptyset),(\setenum{\pmv{a}},\emptyset),(\setenum{\pmv{b}},\emptyset)}
 \gesrel \pmv{c}$$
 namely any context allows to add the event.
 In Figure~\ref{fig:mia-examples}(a) the configurations of this \Ges\ and how they are related are
 depicted.
 At the configuration $\setenum{\pmv{a}}$ it is not possible to add the event $\pmv{b}$ as 
 the entry for $\pmv{b}$ has as context $\setenum{\pmv{a},\pmv{c}}$, hence the pair
 selected is $(\setenum{\pmv{a}},\setenum{\pmv{c}})$ but 
 $\setenum{\pmv{c}}\not\subseteq \setenum{\pmv{a}}$: the event $\pmv{b}$ is not enabled, but it will
 eventually be enabled again when $\pmv{c}$ will be executed (at the 
 configuration $\setenum{\pmv{a},\pmv{c}}$).
\end{example}

\begin{example}[Removing dependencies]\label{ex:shr-dep}
 Consider three events $\pmv{a}, \pmv{b}$ and $\pmv{c}$, and assume that $\pmv{c}$ depends on 
 $\pmv{a}$ unless the event $\pmv{b}$ has occurred, and in this case this dependency is removed. 
 Thus there is a classic causality 
 between $\pmv{a}$ and $\pmv{c}$, but it can be dropped if $\pmv{b}$ occurs. Clearly $\pmv{a}$ and
 $\pmv{b}$ are always enabled. The entries of the \cd-relation are 
 $\setenum{(\emptyset,\emptyset)}\gesrel \pmv{a}$, 
 $\setenum{(\emptyset,\emptyset)}\gesrel \pmv{b}$ and
 $\setenum{(\emptyset,\setenum{\pmv{a}}),(\setenum{\pmv{b}},\emptyset)}\gesrel \pmv{c}$.
 The event $\pmv{c}$ is enabled at the configuration $\setenum{\pmv{b}}$ as 
 the entry for $\pmv{c}$ has as context $\setenum{\pmv{b}}$ and the element
 selected of this entry is $(\setenum{\pmv{b}},\emptyset)$. At the configuration 
 $\emptyset$ the event $\pmv{c}$ is not enabled as the pair selected would be
 $(\emptyset,\setenum{\pmv{a}})$ and 
 $\setenum{\pmv{a}}\not\subseteq \emptyset$.
 The configurations of this \Ges\ are in Figure~\ref{fig:mia-examples}(b).
\end{example}

\begin{example}[Or causality]\label{ex:or-caus}
 Consider three events $\pmv{a}$, $\pmv{b}$ and $\pmv{c}$, assume that $\pmv{a}$ and $\pmv{b}$ are in
 conflict and that $\pmv{c}$ depends on either $\pmv{a}$ or $\pmv{b}$.
 The \cd-relation is 
 $\setenum{(\emptyset,\emptyset)}\gesrel \pmv{a}$, 
 $\setenum{(\emptyset,\emptyset)}\gesrel \pmv{b}$, and
 $\setenum{(\setenum{\pmv{a}},\emptyset),(\setenum{\pmv{b}},\emptyset)}\gesrel \pmv{c}$.
 The configurations are those shown in Figure~\ref{fig:mia-examples}(c).
\end{example}

\begin{example}[Asymmetric conflict]\label{ex:as-conf}
 Consider two events $\pmv{a}$ and $\pmv{b}$, and assume that $\pmv{a}$ is in asymmetric conflict
 with $\pmv{b}$, \emph{i.e.} once that $\pmv{b}$ has happened, then event $\pmv{a}$ cannot be added, but 
 the vice versa is allowed.
 The \cd-relation is 
 $\setenum{(\emptyset,\emptyset),(\setenum{\pmv{b}},\setenum{\pmv{a}})}\gesrel \pmv{a}$ and
 $\setenum{(\emptyset,\emptyset)}\gesrel \pmv{b}$.
 The configurations are those shown in Figure~\ref{fig:mia-examples}(d).
 From the configuration $\setenum{\pmv{b}}$ it not is possible to add the event $\pmv{a}$ as 
 the entry for $\pmv{a}$ has as context $\setenum{\pmv{b}}$, the pair
 selected 
 is $(\setenum{\pmv{b}},\setenum{\pmv{a}})$ but at this configuration 
 the event $\pmv{a}$ is not present.
\end{example}

\begin{example}[Adding dependencies]\label{ex:grow-dep}
 Consider three events $\pmv{a}, \pmv{b}$ and $\pmv{c}$, and assume that $\pmv{c}$ depends on 
 $\pmv{a}$ just when the event $\pmv{b}$ has occurred, and in this case this dependency is added,
 otherwise $\pmv{c}$ may happen without depending on any other event.
 Thus the dependency
 between $\pmv{a}$ and $\pmv{c}$ is added if $\pmv{b}$ occurs. Again $\pmv{a}$ and
 $\pmv{b}$ are always enabled. The \cd-relation is 
 $\setenum{(\emptyset,\emptyset)}\gesrel \pmv{a}$, 
 $\setenum{(\emptyset,\emptyset)}\gesrel \pmv{b}$ and
 $\setenum{(\emptyset,\emptyset),(\setenum{\pmv{b}},\setenum{\pmv{a}})}\gesrel \pmv{c}$.
 The configurations are those displayed in Figure~\ref{fig:mia-examples}(e).
 At the configuration $\setenum{\pmv{b}}$ it not is possible to add the event $\pmv{c}$ as 
 the entry for $\pmv{c}$ has as context $\setenum{\pmv{b}}$, hence the element
 selected at the configuration $\setenum{\pmv{b}}$ 
 is $(\setenum{\pmv{b}},\setenum{\pmv{a}})$ and
 the event $\pmv{a}$ is not present as it should.
\end{example}

\begin{example}[Ternary conflict]\label{ex:new3conf}
 Consider three events $\pmv{a}, \pmv{b}$ and $\pmv{c}$. All the events are 
 singularly enabled and they are in a \emph{ternary} conflict, not a binary one.
 For the events $\pmv{a}$, $\pmv{b}$ and $\pmv{c}$ we stipulate 
 $\setenum{(\emptyset,\emptyset),(\setenum{\pmv{b},\pmv{c}},\setenum{\pmv{a}})}
 \gesrel \pmv{a}$
 $\setenum{(\emptyset,\emptyset),(\setenum{\pmv{a},\pmv{c}},\setenum{\pmv{b}})}
 \gesrel \pmv{b}$
 and
 $\setenum{(\emptyset,\emptyset),(\setenum{\pmv{a},\pmv{b}},\setenum{\pmv{c}})}
 \gesrel \pmv{c}$
 that should be interpreted as follows: if the context is $\emptyset$ then
 the event is enabled without any further condition (the $Y$ are the empty set), if the context
 is a subset of events with two elements then the event is not enabled.
 In Figure~\ref{fig:mia-examples}(f) the configurations of this \Ges\ and how they are related are
 depicted. 
\end{example}

\begin{remark}
 The way context and dependencies are represented in many case is not unique, as we will also see 
 in the remaining of the paper. 
 Consider for instance the example on the Or causality. The entry
 for $\pmv{c}$ is $\setenum{(\setenum{\pmv{a}},\emptyset),(\setenum{\pmv{b}},\emptyset)}\gesrel \pmv{c}$,
 but here the events on which $\pmv{c}$ depends are used as context. An alternative could have 
 been 
 $\setenum{(\setenum{\pmv{a}},\setenum{\pmv{a}}),(\setenum{\pmv{b}},\setenum{\pmv{b}})}\gesrel \pmv{c}$
 where the fact that $\pmv{c}$ depends  either on $\pmv{a}$ or $\pmv{b}$ is made more explicit.
\end{remark}

%
%
We end the presentation of this new kind of event structure adapting the notion of 
\emph{fullness} and \emph{faithfulness} of an event structure with a conflict relation 
(\cite{Bou:FES,AKPN:lmcs18}). Given a \Ges\  $\mathit{E}= (\event{E}, \#, \gesrel)$, and two events 
$\pmv{e}$ and $\pmv{e}'$, we say that that they are in \emph{semantic} conflict if
there is no configuration containing both, \emph{i.e.} $\forall C\in\Conf{\mathit{E}}{\Ges}$.\
$\setenum{\pmv{e},\pmv{e}'}\not\subseteq C$. This will denoted with
$\pmv{e}\sharphatbin\pmv{e}'$.

\begin{definition}\label{de:full-or-faithfull}
 Let $\mathit{E} = (\event{E}, \#, \gesrel)$  be a \Ges.
 Then 
 \begin{enumerate}
   \item $\mathit{E}$ is said to be \emph{full} if $\forall \pmv{e}\in\event{E}$. 
         $\exists C\in\Conf{\mathit{E}}{\Ges}$ such that $\pmv{e}\in C$, and
   \item $\mathit{E}$ is said to be \emph{faithful} if $\forall \pmv{e},\pmv{e}\in\event{E}$. 
         if $\pmv{e}\sharphatbin\pmv{e}'$ then $\pmv{e}\sharpbin\pmv{e}'$.   
 \end{enumerate}           
\end{definition}
Thus fullness means that each event is possible, namely it will executed, whereas
faithfulness implies that each symmetric conflict among two events is represented by the
conflict relation.
\begin{proposition}
 There exists a \Ges\ $\mathit{E}$ that is not full or faithful.
\end{proposition}
\begin{proof}
 For fullness consider the \Ges\ 
 $\mathit{E} = (\setenum{\pmv{e}}, \emptyset, \emptyset)$,
 The unique configuration is $\emptyset$. 
 For faithfulness consider 
 $\mathit{E} = (\setenum{\pmv{e},\pmv{e'},\pmv{e}''}, \pmv{e}\sharpbin\pmv{e'}, \setenum{\setenum{(\emptyset,\emptyset)}\gesrel \pmv{e}, \setenum{(\emptyset,\emptyset)}\gesrel \pmv{e}', 
 \setenum{(\emptyset,\setenum{\pmv{e}})}\gesrel \pmv{e}''})$. Clearly $\pmv{e}'$ and $\pmv{e}''$ are in
 conflict but this does not appear in the conflict relation.
\end{proof}
The notions of fullness can be adapted to each kind of event structure, the one of faithfulness just for
those where a conflict relation is defined.


\section{Event structures}\label{sec:es}
We have introduced a new notion of event structure that we should confront
with the others presented in literature (or at least some of them). 
Therefore we review some of the various definitions of event structures that appeared in literature. 

\subsection{Prime event structures:}\ 
Prime event structures are one among the first proposed and the most widely
studied (\cite{Win:ES}), especially for the connections with \emph{prime algebraic domains} and 
\emph{occurrence nets}.
The dependencies among events are modeled using a \emph{partial order} relation, the 
incompatibility among events is modeled using a symmetric and irreflexive relation, 
the conflict relation, and
it is required that the conflict relation is inherited along the partial order. 
\begin{definition}\label{de:pes-winskel}
 A 
 \emph{prime event structure ({\pes})} is a triple $\mathit{P} = (\event{E}, \leq, \#)$, where 
     \begin{itemize}
       \item 
        $\event{E}$ is a set of \emph{events},
       \item
        $\leq\ \subseteq \event{E}\times \event{E}$ is a 
        well founded \emph{partial order} called \emph{causality relation},
       \item 
        $\forall \pmv{e}\in \event{E}$ the set $\hist{e} = \setcomp{e'}{e'\leq e}$ is 
        finite, and    
       \item
        $\mathord{\#} \subseteq \event{E}\times \event{E}$ is an irreflexive and symmetric relation, called 
        \emph{conflict relation}, such that 
         \begin{itemize}
           \item 
             $\pmv{e}\sharpbin\pmv{e}'\leq \pmv{e}''\ \Rightarrow\ \pmv{e}\sharpbin\pmv{e}''$, and
           \item 
            $\leq \cap\ \# = \emptyset$. 
         \end{itemize}
      \end{itemize}       
\end{definition}
Observe that the finiteness of $\hist{e}$ for each event $\pmv{e}\in \event{E}$ implies 
well foundedness of the partial order relation. 
\begin{definition}\label{de:conf-pes}
 Let $\mathit{P} = (\event{E}, \leq, \#)$ be a \pes\ and $C\subseteq \event{E}$ be a subset of
 events. $C$ is a \emph{configuration} of $\mathit{P}$ iff $\CF{C}$ 
 and for each $\pmv{e}\in C$ it holds that $\hist{\pmv{e}}\subseteq C$. 
\end{definition}
The set of configuration of a \pes\ is denoted with 
$\Conf{\mathit{P}}{\pes}$. 
Clearly $(\Conf{\mathit{P}}{\pes}, \subseteq)$ is a partial order. 

\begin{definition}\label{de:mapsto-pes}
Let $\mathit{P} = (\event{E}, \leq, \#)$ be a \pes.
With $\mapsto_{\pes}$ we denote the relation 
over $\Conf{\mathit{P}}{\pes} \times \Conf{\mathit{P}}{\pes}$ defined as $C \mapsto_{\pes} C'$ iff 
$C\subset C'$ and $C' = C\cup\setenum{\pmv{e}}$
for some $\pmv{e}\in \event{E}$. 
\end{definition}

\subsection{Flow event structure:}\
In \emph{flow} event structures, introduced in \cite{Bou:FES}, the dependency and the 
conflict relations have little constraints.

\begin{definition}\label{de:flow-es}
 A \emph{flow} prime event structure (\emph{{\fes}}) 
 is a triple 
$\mathit{F} = (\event{E},\prec, \#)$, where 
      \begin{itemize}
       \item 
        $\event{E}$ is a set of \emph{events},
       \item
        $\prec\ \subseteq \event{E}\times \event{E}$ is an irreflexive relation called 
        the  \emph{flow relation}, and 
       \item
        $\#\subseteq \event{E}\times \event{E}$ is a symmetric \emph{conflict relation}.
      \end{itemize}  
\end{definition}
The flow relation is the one representing the dependencies among events, 
whereas the conflict relation is required
to be just symmetric, which means that an event $\pmv{e}$ could be in conflict with itself, meaning
that it cannot be executed. 
The dependencies are not so \emph{prescriptive} as in the case of \pes, as we will see in the 
definition of configuration.
Indeed, the notion of configuration plays a major role, as it is on this level that the conditions 
on the relations are placed.

\begin{definition}\label{de:conf-fes}
 Let $\mathit{F} = (\event{E}, \prec, \#)$ be a \fes\ and $C\subseteq \event{E}$ be a subset of
 events. $C$ is a \emph{configuration} of $\mathit{F}$ iff 
  \begin{itemize}
    \item $\CF{C}$, 
    \item for each $\pmv{e}\in C$, for each $\pmv{e}'\prec \pmv{e}$ either
          $\pmv{e}'\in C$ or there exists $\pmv{e}''\in C$ such that 
          $\pmv{e}''\prec \pmv{e}$ and $\pmv{e}'\sharpbin\pmv{e}''$, and
    \item $\prec^{\ast}\cap\ C\times C$ is a partial order.
  \end{itemize}           
\end{definition}
The intuition is that, for each event $\pmv{e}$ in the configuration, if an event 
which should be present in the configuration because it is \emph{suggested} by the flow
relation is not present then there is another event suggested by the flow relation
which conflicts with the absent one.
The set of configurations of a \fes\ $\mathit{F}$ is denoted with 
$\Conf{\mathit{F}}{\fes}$.
Again in \cite{Bou:FES} it is shown that 
$(\Conf{\mathit{F}}{\fes}, \subseteq)$ is not only a partial order but also a prime algebraic domain,
which establish a clear connection between \fes\ and \pes.

Also on these configurations we can state how to reach a configuration from another just adding an event.
\begin{definition}\label{de:mapsto-fes}
Let $\mathit{F} = (\event{E}, \prec, \#)$ be a \fes.
With $\mapsto_{\fes}$ we denote the relation 
over $\Conf{\mathit{F}}{\fes} \times \Conf{\mathit{F}}{\fes}$ defined as $C \mapsto_{\fes} C'$ iff 
$C\subset C'$ and $C' = C\cup\setenum{\pmv{e}}$
for some $\pmv{e}\in \event{E}$.
\end{definition} 

\subsection{Relaxed prime event structures:}\  
Some of the requirements of a 
\pes, the one on the dependencies among events (here called enabling) and the one regarding the 
conflicts among events
(which does not need to be saturated), can be relaxed yielding a \emph{relaxed} prime event structure
 (\cite{AKPN:DC,AKPN:lmcs18}). 
 This notion is introduced to allow a clear notion of dynamic causality and the idea is to focus on the 
 dependencies and conflicts that are somehow the \emph{generators} of the causal dependency
 relation and of the conflict relation of a \pes.
 In this definition the events that must be \emph{present} in a state to allow the execution of
 another one are the events in a (finite) subset called \emph{immediate causes} and often denoted
 with $\ic$. 
\begin{definition}\label{de:rpes}
A \emph{relaxed} prime event structure (\emph{{r\pes}}) 
 is a triple 
$(\event{E},\to, \#)$, where 
      \begin{itemize}
       \item 
        $\event{E}$ is a set of \emph{events},
       \item
        $\to\ \subseteq \event{E}\times \event{E}$ is the  \emph{enabling relation} such that
        $\forall \pmv{e}\in\event{E}$ the set $\ic(\pmv{e}) = \setcomp{\pmv{e'}}{\pmv{e'}\to\pmv{e}}$
        is finite, 
       \item 
         for every $\pmv{e}\in \event{E}$. $\to^{+} \cap \ic(\pmv{e}) \times \ic(\pmv{e})$ 
         is irreflexive, and 
       \item
        $\#\subseteq \event{E}\times \event{E}$ is an irreflexive and symmetric \emph{conflict relation}.
      \end{itemize}  
\end{definition}
The intuition is that the $\to$ relation plays the role of the causality relation and the 
conflict relation models conflicts among events, as before.
The immediate causes can be seen as a mapping $\ic\colon \event{E} \to \Powfin{\event{E}}$.

The above definition of relaxed prime event structure is slightly different from the one in
\cite{AKPN:DC} and \cite{AKPN:lmcs18}. There $\to$ is a relation without any further requirement
as it is done here. The requirement that $\to^{\ast} \cap \ic(\pmv{e}) \times \ic(\pmv{e})$ 
is there moved to the 
definition of configuration whereas the one that an event should have finite causes is here introduced
in the definition as we consider configurations without further constraint.
\begin{definition}\label{de:rpes-conf}
Let $\mathit{T} = (\event{E},\to, \#)$ be a r\pes\ and let $C\subseteq \event{E}$ be a subset of events. 
We say that $C$ is a \emph{configuration} of the r\pes\ $\mathit{T}$ iff  
\begin{itemize}
 \item
   $\CF{C}$, and 
 \item 
   for each $\pmv{e}\in C$. $\ic(\pmv{e}) \subseteq C$. 
\end{itemize} 
\end{definition}
The set of configuration of a r\pes\ is denoted with 
$\Conf{\mathit{T}}{r\pes}$.
We prove that this definition of  configuration implies that there is no causality cycle.
\begin{proposition}\label{pr:conf-rpes-are-po}
Let $\mathit{T} = (\event{E},\to, \#)$ be a r\pes\ and consider $C \in \Conf{\mathit{T}}{r\pes}$.
Then $\to^{\ast}\cap\ C\times C$ is a partial order.
\end{proposition}
\begin{proof}
Assume it is not the case. Then, for $\pmv{e}, \pmv{e}'\in C$, we have  
$\pmv{e} \to^{\ast} \pmv{e}'$ and $\pmv{e}' \to^{\ast} \pmv{e}$ with $\pmv{e} \neq \pmv{e}'$. 
But on each $\pmv{e}\in C$ we have that $\to^{+}$ is irreflexive on $\ic(\pmv{e})$
which contradicts the fact that $\to^{\ast}$ is a partial order on the configuration $C$.
\end{proof}

Similarly to \pes\ and \fes, we have that also 
$(\Conf{\mathit{T}}{r\pes}, \subseteq)$ is a partial order. 

As it should be expected, \pes\ and r\pes\ are related.
A \pes\ is also a r\pes: the causality relation is the enabling relation and the conflict relation
is the same one. $\pmv{e}$ is added to a configuration $C$ when its causes are
in $C$ and no conflict arises. 
For the vice versa, given a full r\pes\ $\mathit{T} = (\event{E},\to, \#)$, it is not difficult to see
that $(\event{E},\to^{\ast}, \hat{\#})$ is a \pes, where $\to^{\ast}$ is the reflexive and 
transitive closure
of $\to$ and $\hat{\#}$ is obtained by $\#$ stipulating that $\# \subseteq \hat{\#}$ and 
it is closed with respect to $\to^{\ast}$, \emph{i.e.} if 
$\pmv{e}\sharphatbin\pmv{e}' \to^{\ast} \pmv{e}''$
then $\pmv{e}\sharphatbin\pmv{e}''$. 
Indeed, the fact that $\to^{\ast}$ is a partial order is guaranteed by the fact that each 
event is executable, that $\to^{\ast}$ is well founded is implied by the finiteness of causes for each 
event $\pmv{e}\in\event{E}$ and 
$\hat{\#}$ is the semantic closure of $\#$: no new conflict is introduced. 
In case the r\pes\ is not full, we have to focus on the events appearing in a configuration. 

\begin{definition}\label{de:mapsto-rpes}
Let $\mathit{T} = (\event{E},\to, \#)$ be a r\pes.
With $\mapsto_{r\pes}$ we denote the relation 
over $\Conf{\mathit{T}}{r\pes} \times \Conf{\mathit{T}}{r\pes}$ defined as $C \mapsto_{r\pes} C'$ iff 
$C\subset C'$ and $C' = C\cup\setenum{\pmv{e}}$
for some $\pmv{e}\in \event{E}$. 
\end{definition}

\subsection{Dynamic causality event structures:}\ 
We now review a notion of event structure where causality may change (\cite{AKPN:DC,AKPN:lmcs18}). 
The idea is to enrich a r\pes\ with two relations, one modeling the 
shrinking causality (some dependencies are dropped) and the other the growing causality (some
dependencies are added).
The shrinking and the growing causality relations are ternary relations stipulating that the
happening of a specific event (the \emph{modifier}) allows to drop or add a specific cause 
(the contribution) for another event (the \emph{target}).

\begin{definition}\label{de:shr-rel}
 Let $\mathit{T} = (\event{E},\to, \#)$ be a r\pes. 
 A \emph{shrinking causality} relation is a ternary relation
 $\shrinkt\ \subseteq \event{E}\times\event{E}\times\event{E}$, whose elements 
 are denoted with $\shrink{\pmv{e}}{\pmv{e}''}{\pmv{e}'}$, such that
 \begin{itemize}
  \item $\setenum{\pmv{e}'} \cap \setenum{\pmv{e},\pmv{e}''} = \emptyset$, and
  \item $\pmv{e}\to \pmv{e}''$.
 \end{itemize}
 We say that $\shrinkt$ is a shrinking relation with respect to the enabling relation $\to$.
\end{definition}
Given $\shrink{\pmv{e}}{\pmv{e}''}{\pmv{e}'}$, 
$\pmv{e}'$ is called \emph{modifier}, $\pmv{e}''$ \emph{target} and
$\pmv{e}$ \emph{contribution}.
To drop an enabling we require that the enabling is present.

Associated to this relation we introduce a number of auxiliary subsets of events.
\begin{definition}\label{de:aux-shr-set}
Let $\shrinkt\ \subseteq \event{E}\times\event{E}\times\event{E}$ be a shrinking causality
relation with respect to an enabling relation $\to$. Then
\begin{enumerate}
  \item $\shrset{\pmv{e}''} = \setcomp{\pmv{e}'}{\shrink{\pmv{e}}{\pmv{e}''}{\pmv{e}'}}$ 
        is the set of modifiers for a given target $\pmv{e}''$, 
 \item $\drop{\pmv{e}'}{\pmv{e}''} =  \setcomp{\pmv{e}}{\shrink{\pmv{e}}{\pmv{e}''}{\pmv{e}'}}$
       is the set of contributions for a given modifier $\pmv{e}'$ and a given 
       target $\pmv{e}''$, and 
 \item $\dc(H,\pmv{e}'') = \bigcup_{\pmv{e'}\in H\cap\shrset{\pmv{e}}} \drop{\pmv{e}'}{\pmv{e}''}$
       is the set of \emph{dropped} causes, with respect to a subset of events $H\subseteq\event{E}$, 
       for the event $\pmv{e}$.
\end{enumerate}
\end{definition}
$\shrset{\pmv{e}''}$ is the set of all possible modifiers associated to the event 
$\pmv{e}''$, thus it contains the set of events that may change the dependencies for it,
and for each of them the set
$\drop{\pmv{e}'}{\pmv{e}''}$ contains the \emph{dropped} causes. Finally,
for a given event $\pmv{e}''$, the set $\dc(H,\pmv{e})$ contains the dropped causes
provided that the events in $H$ are already been executed.

\begin{definition}\label{de:gro-rel}
 Let $\mathit{T} = (\event{E},\to, \#)$ be a r\pes. 
 A  \emph{growing causality} relation is a ternary relation 
 $\growt\ \subseteq \event{E}\times\event{E}\times\event{E}$, whose elements 
 are denoted as $\grow{\pmv{e}'}{\pmv{e}}{\pmv{e}''}$, such that
  \begin{itemize}
  \item $\setenum{\pmv{e}'} \cap \setenum{\pmv{e},\pmv{e}''} = \emptyset$, and
  \item $\neg(\pmv{e}\to \pmv{e}'')$.
 \end{itemize}
 We say that $\growt$ is a growing relation with respect to the enabling relation $\to$.
\end{definition}
Given $\grow{\pmv{e}'}{\pmv{e}}{\pmv{e}''}$, 
$\pmv{e}'$ is called \emph{modifier}, $\pmv{e}''$ \emph{target} and
$\pmv{e}$ \emph{contribution}. To add a dependency among two events this dependency must be absent.

\begin{definition}\label{de:aux-grow-set}
Let $\growt\ \subseteq \event{E}\times\event{E}\times\event{E}$ be a growing causality
relation with respect to an enabling relation $\to$. Then
\begin{enumerate}
  \item $\groset{\pmv{e}''} = \setcomp{\pmv{e}'}{\grow{\pmv{e}'}{\pmv{e}}{\pmv{e}''}}$ is the set
        of modifiers for a given target $\pmv{e}''$,
  \item $\agg{\pmv{e}'}{\pmv{e}''} =  \setcomp{\pmv{e}}{\grow{\pmv{e}'}{\pmv{e}}{\pmv{e}''}}$
        is the set of contributions for a given modifier $\pmv{e}'$ and a given target 
        $\pmv{e}''$, and 
  \item $\ac(H,\pmv{e}) = \bigcup_{\pmv{e'}\in H\cap\groset{\pmv{e}}}  \agg{\pmv{e}'}{\pmv{e}}$
        is the set of \emph{added} causes, with respect to a subset of events $H\subseteq \event{E}$, 
        for the event $\pmv{e}$.
\end{enumerate}
\end{definition}

The two relations of shrinking and growing causality give the 
functions $\dc\colon \Powfin{\event{E}}\times \event{E} \to \Powfin{\event{E}}$. 
and $\ac\colon \Powfin{\event{E}}\times \event{E} \to \Powfin{\event{E}}$. 

\begin{definition}
\label{de:dces}
A \emph{dynamic causality event structure ({\dces})} is a quintuple
$\mathit{D} = (\event{E}, \to, \#, \shrinkt, \growt)$, where $(\event{E}, \to, \#)$ is a r\pes,
$\shrinkt\ \subseteq \event{E}\times\event{E}\times\event{E}$ is the \emph{shrinking causality} relation,
$\growt\ \subseteq \event{E}\times\event{E}\times\event{E}$ is the \emph{growing causality} relation, and
are such that for all $\pmv{e}, \pmv{e}', \pmv{e}'' \in \event{E}$
\begin{enumerate}
\item\label{dces-cond-1} 
  $\shrink{\pmv{e}}{\pmv{e}''}{\pmv{e}'} \wedge \nexists \pmv{e}''' \in
  \event{E}.\ \grow{\pmv{e}'''}{\pmv{e}}{\pmv{e}''}\Longrightarrow \pmv{e}\rightarrow \pmv{e}''$,
\item\label{dces-cond-2}  
  $\grow{\pmv{e}'}{\pmv{e}}{\pmv{e}''} \wedge \nexists \pmv{e}''' \in \event{E}. 
  \shrink{\pmv{e}}{\pmv{e}''}{\pmv{e}'''} \Longrightarrow \neg(\pmv{e}\rightarrow \pmv{e}'')$,
\item\label{dces-cond-3}  
  $\grow{\pmv{e}'}{\pmv{e}}{\pmv{e}''} \Longrightarrow \neg(\shrink{\pmv{e}}{\pmv{e}''}{\pmv{e}'})$, and
\item\label{dces-order}
  $\forall \pmv{e}, \pmv{e}'\in \event{E}.\ \nexists \pmv{e}'', \pmv{e}'''\in \event{E}.\ 
  \shrink{\pmv{e}}{\pmv{e}'}{\pmv{e}''}$ and $\grow{\pmv{e}'''}{\pmv{e}}{\pmv{e}'}$.
\end{enumerate}
\end{definition}
For further comments on this definition we refer to \cite{AKPN:DC} and \cite{AKPN:lmcs18}.
It should be observed, however, that
the definition we consider here is slightly 
less general of the one
presented there, as we add a further condition, the last one, 
which is defined in \cite{AKPN:DC-rep} and \cite{AKPN:lmcs18}, and  
does not allow that the same contribution can be added and removed by two different modifiers. 
These are called in \cite{AKPN:DC-rep} \emph{single state dynamic causality event structures} and rule
out the fact that some causality (or absence of) depends on the order of modifiers.
Conditions~1  and~2 simply state that in the case of the shrinking relation the dependency should 
be present, and in the case of the
growing the dependency should be absent; 
condition~3 says that if a dependency is added then it cannot be removed, or
a removed dependency cannot be added, and
the final condition express the fact that two modifiers, one growing and
the other shrinking,  cannot act on the same dependency.
Clearly a \dces\ where $\shrinkt$ and $\growt$ are empty is a r\pes.

\begin{definition}\label{de:dces-conf}
Let $\mathit{D} = (\event{E}, \to, \#, \shrinkt, \growt)$ be a \dces. 
Let $C$ be a subset of $\event{E}$.
We say that $C$ is a \emph{configuration} of the \dces\ 
iff there exists a 
sequence of distinct events 
$\rho = \pmv{e}_1\cdots \pmv{e}_n\cdots$ over $\event{E}$ 
such that  
  \begin{itemize}
      \item
        $\overline{\rho} = C$,
     \item 
        $\CF{\overline{\rho}}$, and 
    \item
       $\forall 1 \leq i \leq \len{\rho}.\ 
       ((\ic(\pmv{e}_i) \cup \ac(\overline{\rho_{i-1}}, \pmv{e}_i)) \setminus 
       \dc(\overline{\rho_{i-1}}, \pmv{e}_i))
       \subseteq \overline{\rho_{i-1}}$.
  \end{itemize}     
The set of configuration of a \dces\ is denoted with 
$\Conf{\mathit{D}}{\dces}$.
\end{definition}
To add an event to a subset of events we have to check that the added causes are
present ($\dc(\overline{\rho_{i-1}}, \pmv{e}_i)$) whereas the dropped causes
($\ac(\overline{\rho_{i-1}}, \pmv{e}_i)$) can be ignored, \emph{i.e.} they do not have to
be present.
 
The following definition introduce the obvious relation between configurations of a \dces.
\begin{definition}\label{de:mapsto-dces}
Let $\mathit{D} = (\event{E}, \to, \#, \shrinkt, \growt)$ be a \dces.
With $\mapsto_{\dces}$ we denote the relation 
over $\Conf{\mathit{D}}{\dces} \times \Conf{\mathit{D}}{\dces}$ defined as $C \mapsto_{\dces} C'$ iff 
$C\subset C'$, $C' = C\cup\setenum{\pmv{e}}$ for some $\pmv{e}\in \event{E}$ and
$((\ic(\pmv{e}) \cup \ac(C, \pmv{e})) \setminus \dc(C, \pmv{e}))\subseteq C$.
\end{definition}

 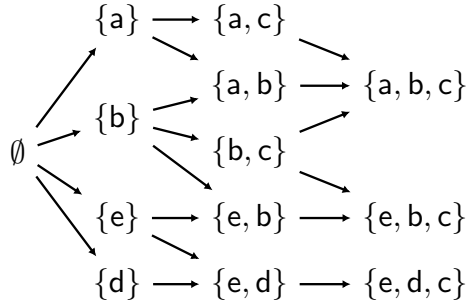
\begin{figure}[!ht]
 \centerline{\scalebox{1.1}{\begin{tikzpicture}
[bend angle=45, scale=0.8, 
  read/.style={-,shorten
    <=0pt,thick},
  pre/.style={<-,shorten
    <=0pt,>=stealth,>={Latex[width=1mm,length=1mm]},thick}, post/.style={->,shorten
    >=0,>=stealth,>={Latex[width=1mm,length=1mm]},thick},place/.style={circle, draw=black,
    thick,minimum size=4mm}, transition/.style={rectangle, draw=black!0,
    thick,minimum size=4mm}, invplace/.style={circle,
    draw=black!0,thick}]
\node[transition] (spe) at (0,2.5) {};  
\node[transition] (in) at (0,0) {$\emptyset$};
\node[transition] (a) at (1.5,2) {$\setenum{\pmv{a}}$}
edge[pre] (in);
\node[transition] (e) at (1.5,-1) {$\setenum{\pmv{e}}$}
edge[pre] (in);
\node[transition] (b) at (1.5,0.5) {$\setenum{\pmv{b}}$}
edge[pre] (in);
\node[transition] (d) at (1.5,-2) {$\setenum{\pmv{d}}$}
edge[pre] (in);
\node[transition] (ac) at (3.5,2) {$\setenum{\pmv{a},\pmv{c}}$}
edge[pre] (a);
\node[transition] (ab) at (3.5,1) {$\setenum{\pmv{a},\pmv{b}}$}
edge[pre] (a)
edge[pre] (b);
\node[transition] (bc) at (3.5,0) {$\setenum{\pmv{b},\pmv{c}}$}
edge[pre] (b);
\node[transition] (eb) at (3.5,-1) {$\setenum{\pmv{e},\pmv{b}}$}
edge[pre] (b)
edge[pre] (e);
\node[transition] (ed) at (3.5,-2) {$\setenum{\pmv{e},\pmv{d}}$}
edge[pre] (d)
edge[pre] (e);
\node[transition] (abc) at (6,1) {$\setenum{\pmv{a},\pmv{b},\pmv{c}}$}
edge[pre] (ab)
edge[pre] (ac)
edge[pre] (bc);
\node[transition] (ebc) at (6,-1) {$\setenum{\pmv{e},\pmv{b},\pmv{c}}$}
edge[pre] (eb)
edge[pre] (bc);
\node[transition] (edc) at (6,-2) {$\setenum{\pmv{e},\pmv{d},\pmv{c}}$}
edge[pre] (ed);
\end{tikzpicture}}}
 \caption{The configurations of the \dces\ of the Example~\ref{ex:dces}}\label{fig:dces}
 \end{figure}
 
 \begin{example}\label{ex:dces}
 Consider the set of events $\setenum{\pmv{a}, \pmv{b}, \pmv{c}, \pmv{d}, \pmv{e}}$,
 with $\pmv{b}\to \pmv{c}$,  $\shrink{\pmv{b}}{\pmv{c}}{\pmv{a}}$,
 $\grow{\pmv{d}}{\pmv{e}}{\pmv{c}}$, $\pmv{a}\sharpbin\pmv{e}$ and
 $\pmv{d}\sharpbin\pmv{b}$. $\pmv{a}$ and $\pmv{d}$  are the modifiers
 for the target $\pmv{c}$, the happening of $\pmv{a}$ has the effect that the cause $\pmv{b}$ may be dropped,
 and the one of $\pmv{d}$ that the cause $\pmv{e}$ should be added for $\pmv{c}$. 
 If the prefix of the trace is $\pmv{b}\pmv{c}$ (the target $\pmv{c}$ is executed
 before of one of its modifiers $\pmv{a}$ and $\pmv{d}$) then the final part of the trace
 is any either $\pmv{a}$ or $\pmv{e}$, and as $\pmv{d}\sharpbin\pmv{b}$ we have that $\pmv{d}$ cannot be added. 
 If the modifier $\pmv{a}$ is executed before $\pmv{c}$ then we have the traces
 $\pmv{a}\pmv{c}$ (as the immediate cause $\pmv{b}$ of $\pmv{c}$ is dropped by 
 $\pmv{a}$) followed by $\pmv{b}$ or $\pmv{d}$, 
 and if the modifier $\pmv{d}$ is executed, then before adding $\pmv{c}$, we need
 $\pmv{e}$ (the modifier $\pmv{d}$ add the immediate cause $\pmv{e}$ for $\pmv{c}$), 
 and in this case we cannot add $\pmv{b}$ for sure as it is in conflict with $\pmv{d}$
 or $\pmv{a}$ as it is in conflict with $\pmv{e}$.
 If both modifiers $\pmv{a}$ and $\pmv{d}$ happen, 
 then the event $\pmv{c}$ is permanently disabled, as it needs the
 contribution $\pmv{e}$ (growing cause) which is in conflict with $\pmv{a}$.
 In Figure~\ref{fig:dces} are shown the configurations of this \dces\ and the 
 $\mapsto_{\dces}$ relation.
\end{example}
 A shrinking event structure (\ses) is a \dces\ where the $\growt$ relation is empty and a growing
 event structure (\ges) is a \dces\ where the $\shrinkt$ relation is empty. 
 
 In \cite{AKPN:lmcs18} it is shown that many kind of event structures can be seen as 
 \dces, like \emph{bundle} event structure \cite{Lan92:BES} or \emph{dual} event structure 
 \cite{LBK:dual}. 

\subsection{Inhibitor event structures:}\
Inhibitor event structures (\cite{BBCP:rivista}) are equipped with a relation
$\inh\ \subseteq \Powone{\event{E}} \times \event{E} \times \Powfin{\event{E}}$ allowing
to model conflicts (even asymmetric) as well as temporary inhibitions. With $\Powone{\event{E}}$
we denote the subsets of events with cardinality at most one (the empty set or singletons).
The intuition behind this relation is the following: given 
$\De{a}{\pmv{e}}{A}$, the event $\pmv{e}$ is enabled at a configuration is whenever the configuration
contains the set $a$, then its intersection with $A$ is non empty. 
Hence the event in a non empty $a$ \emph{inhibits} the happening of  $\pmv{e}$ unless some event
in $A$ has happened as well. 
We stipulate that given $\De{a}{\pmv{e}}{A}$ the events in $A$ are pairwise conflicting 
(denoted with $\#(A)$).
Two events $\pmv{e}$ and $\pmv{e}'$ are in conflict if $\De{\setenum{\pmv{e}'}}{\pmv{e}}{\emptyset}$
and $\De{\setenum{\pmv{e}}}{\pmv{e}'}{\emptyset}$. An \emph{or}-causality relation $<$ is definable stipulating 
that $A < \pmv{e}$ if $\De{\emptyset}{\pmv{e}}{A}$, and that if $A < \pmv{e}$ and $B < \pmv{e}'$ for some
$\pmv{e}'\in A$ then also $B < \pmv{e}$. 
This relation should be interpreted as follows: $A < \pmv{e}$ means that 
if $\pmv{e}$ is present, then also an event in $A$ should be present.

\begin{definition}\label{de:ies}
 An \emph{inhibitor event structure ({\ies})} is a pair $\mathit{I} = (\event{E}, \inh)$, where 
 \begin{itemize}
  \item 
        $\event{E}$ is a set of \emph{events} and
  \item 
        $\inh\ \subseteq \Powone{\event{E}} \times \event{E} \times \Powfin{\event{E}}$ is a 
        relation such that for each $\De{a}{\pmv{e}}{A}$ it holds that $\#(A)$ and
        $a\cup A\neq\emptyset$.
 \end{itemize}       
\end{definition}
We briefly recall the intuition: consider an event $\pmv{e}$ and a triple in the $\inh$ relation
$\De{a}{\pmv{e}}{A}$. Then $\pmv{e}$ can be added provided that if the event in $a$ is present also one in $A$ should be present. 
\begin{definition}\label{de:ies-traces}
 Let $\mathit{I} = (\event{E}, \inh)$ be an \ies. 
 Let $C$ be a subset of $\event{E}$. 
 We say that $C$ is a \emph{configuration} of the \ies\ $\mathit{I}$ iff there exists a 
 sequence of distinct events $\rho = \pmv{e}_1\cdots\pmv{e}_n\cdots$ over $\event{E}$ 
 such that 
 \begin{itemize}
  \item $\overline{\rho} = C$, and 
  \item for each $i\leq  \len{\rho}$, for each $\De{a}{\pmv{e}_i}{A}$, it holds that
        $a \subseteq \overline{\rho_{i-1}}\ \Rightarrow\ \overline{\rho_{i-1}}\cap A\neq \emptyset$.
 \end{itemize}
 The set of configuration of a \ies\ is denoted with 
 $\Conf{\mathit{I}}{\ies}$.
\end{definition} 

\begin{definition}\label{de:mapsto-ies}
Let $\mathit{I} = (\event{E}, \inh)$ be an \ies.
With $\mapsto_{\ies}$ we denote the relation 
over $\Conf{\mathit{I}}{\ies} \times \Conf{\mathit{D}}{\ies}$ defined as $C \mapsto_{\ies} C'$ iff 
$C\subset C'$ and $C' = C\cup\setenum{\pmv{e}}$ for some $\pmv{e}\in \event{E}$.
\end{definition}
\begin{example}\label{ex:ies}
 Consider three events $\pmv{a}, \pmv{b}$ and $\pmv{c}$, $\De{\setenum{\pmv{a}}}{\pmv{c}}{\setenum{\pmv{b}}}$
 and $\De{\emptyset}{\pmv{b}}{\setenum{\pmv{a}}}$. The maximal event traces are
 $\pmv{c}\pmv{a}\pmv{b}$ and $\pmv{a}\pmv{b}\pmv{c}$. The event $\pmv{c}$ is inhibited when the
 event $\pmv{a}$ has occurred unless the event $\pmv{b}$ has occurred as well.
 The configurations are $\emptyset$, $\setenum{\pmv{a}}$, $\setenum{\pmv{c}}$,
 $\setenum{\pmv{a},\pmv{b}}$, $\setenum{\pmv{a},\pmv{c}}$ and $\setenum{\pmv{a},\pmv{b},\pmv{c}}$
 and are reached as follows: $\emptyset \mapsto_{\ies} \setenum{\pmv{a}}$, 
 $\emptyset \mapsto_{\ies} \setenum{\pmv{c}}$, 
 $\setenum{\pmv{a}} \mapsto_{\ies} \setenum{\pmv{a},\pmv{b}}$,
 $\setenum{\pmv{c}} \mapsto_{\ies} \setenum{\pmv{a},\pmv{c}}$, 
 $\setenum{\pmv{a},\pmv{b}}\mapsto_{\ies} \setenum{\pmv{a},\pmv{b},\pmv{c}}$ and
 $\setenum{\pmv{a},\pmv{c}}\mapsto_{\ies} \setenum{\pmv{a},\pmv{b},\pmv{c}}$.
 \begin{figure}[ht]
 \centerline{\scalebox{1.1}{\begin{tikzpicture}
[bend angle=45, scale=0.8, 
  read/.style={-,shorten
    <=0pt,thick},
  pre/.style={<-,shorten
    <=0pt,>=stealth,>={Latex[width=1mm,length=1mm]},thick}, post/.style={->,shorten
    >=0,>=stealth,>={Latex[width=1mm,length=1mm]},thick},place/.style={circle, draw=black,
    thick,minimum size=4mm}, transition/.style={rectangle, draw=black!0,
    thick,minimum size=4mm}, invplace/.style={circle,
    draw=black!0,thick}]
\node[transition] (spe) at (0,1.5) {};  
\node[transition] (in) at (0,0) {$\emptyset$};
\node[transition] (a) at (1.5,1) {$\setenum{\pmv{a}}$}
edge[pre] (in);
\node[transition] (c) at (1.5,-1) {$\setenum{\pmv{c}}$}
edge[pre] (in);
\node[transition] (ab) at (3.5,1) {$\setenum{\pmv{a},\pmv{b}}$}
edge[pre] (a);
\node[transition] (ac) at (3.5,-1) {$\setenum{\pmv{a},\pmv{c}}$}
edge[pre] (c);
\node[transition] (abc) at (6,0) {$\setenum{\pmv{a},\pmv{b},\pmv{c}}$}
edge[pre] (ab)
edge[pre] (ac);
\end{tikzpicture}}}
 \caption{The configurations of the \ies\ in the Example~\ref{ex:ies}}
 \end{figure}
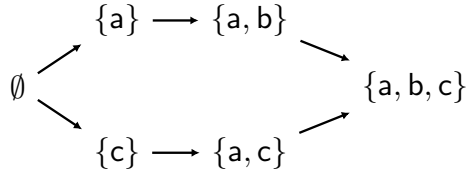
\end{example}

\subsection{Event structures with resolvable conflicts:}\ 
We finally recall the notion of event structure with resolvable conflicts (\cite{GP:ESRC}).

\begin{definition}\label{de:rces}
 An \emph{event structure with resolvable conflicts} (\rces) is the pair $\mathit{R} = (\event{E}, \vdash)$ where
 $\event{E}$ is a set of events and $\mathord{\vdash} \subseteq \Pow{\event{E}}\times \Pow{\event{E}}$ is the \emph{enabling}
 relation.
\end{definition}
No restriction is posed on the enabling relation. The intuition is that stipulating $X \vdash Y$ one state that
for all the events in $Y$ to occur, also the events in the set $X$ should have occurred first.

\begin{definition}\label{de:single-step}
 The \emph{single event} transition relation $\leadsto \subseteq \Pow{\event{E}}\times \Pow{\event{E}}$ of a
 \rces\ $\mathit{R} = (\event{E}, \vdash)$ is given by 
 $X \leadsto Y \Leftrightarrow (X \subseteq Y \land
 \card{Y\setminus X} \leq 1 \land \forall Z\subseteq Y.\ \exists W\subseteq X.\ W\vdash Z)$.
\end{definition}
With this notion it is possible to define what a configuration is: it 
is a subset $X$ of events such that $X\leadsto X$. The requirement that $X\leadsto X$ implies that each 
subset of events is enabled in the configuration.
\begin{definition}\label{de:rces-event-trace}
 Let $\mathit{R} = (\event{E}, \vdash)$ be a \rces. 
 Let $C$ be a subset of $\event{E}$. 
 We say that $C$ is a \emph{configuration} of the \rces\ $\mathit{R}$ iff there exists a 
 sequence of distinct events $\rho = \pmv{e}_1\cdots\pmv{e}_n\cdots\dots$ over $\event{E}$ 
 such that for each $1\leq i\leq \len{\rho}$ it holds that
  \begin{itemize} 
   \item 
   $\overline{\rho_{i-1}}$ and $\overline{\rho_{i}}$ are configurations, and
   \item 
   $\overline{\rho_{i-1}} \leadsto \overline{\rho_{i}}$.
  \end{itemize} 
  The set of configuration of a \rces\ is denoted with 
 $\Conf{\mathit{R}}{\rces}$.
\end{definition}
The enabling relation $\vdash$ is used not only to state under
which condition an event may happen, but also to stipulate when an event is \emph{deducible} from
a set of events, and this justifies why the  deduction symbol is used for this relation.

\begin{definition}\label{de:mapsto-rces}
Let $\mathit{R} = (\event{E}, \vdash)$ be a \rces.
With $\mapsto_{\rces}$ we denote the relation 
over $\Conf{\mathit{R}}{\rces}\times\Conf{\mathit{R}}{\rces}$ defined as
$C \mapsto_{\rces} C'$ iff 
$C\subset C'$ and $C' = C\cup\setenum{\pmv{e}}$
for some $\pmv{e}\in \event{E}$ and $C \leadsto C'$. 
\end{definition}
Observe that $\mapsto_{\rces}$ is, in this case, essentially $\leadsto$.

\begin{example}\label{ex:rces}
 Consider three events $\pmv{a}, \pmv{b}$ and $\pmv{c}$, and $\emptyset\vdash X$ where 
 $X\subseteq \setenum{\pmv{a}, \pmv{b}, \pmv{c}}$ with $X\neq \setenum{\pmv{a}, \pmv{b}}$ and
 $\setenum{\pmv{c}}\vdash \setenum{\pmv{a}, \pmv{b}}$. The intuition is that all the events are 
 singularly enabled but $\pmv{a}$ and $\pmv{b}$ are in conflict unless $\pmv{c}$ has not happened.
 In fact $\setenum{\pmv{a}, \pmv{b}}$ is not a configuration as taking $\setenum{\pmv{a}, \pmv{b}}$ as
 the $Z\subseteq \setenum{\pmv{a}, \pmv{b}}$ of the notion of single event transition relation, 
 there
 is no subset of $\setenum{\pmv{a}, \pmv{b}}$ enabling these two events.
 
 The configurations and how they are reached are those of the Example~\ref{ex:new2res}
 depicted in Figure~\ref{fig:mia-examples}(a).
\end{example}


\section{Embedding and comparing Event Structures}\label{sec:relating}
We now show that each of the event structure we have seen so far can be seen as a \Ges,
and also how to compare them.
For the sake of simplicity, we will consider full event structures, \emph{i.e} 
each event $\pmv{e}$ of the event structure is \emph{executable}, namely that there 
is at least a configuration containing it.

\subsection{Comparing Event Structures:}\ 
We start by devising how we can compare two event structures of any kind. The intuition
is obvious: two event structures are equivalent iff they have the same configurations and 
the $\mapsto$ relations
defined on configurations coincide. 
We recall the notion of \emph{event automaton} (\cite{PP:NEPC}).

\begin{definition}\label{de:event-automata}
 Let $\event{E}$ be a set of \emph{events}. An \emph{event automaton} over $\event{E}$ (\ea) 
 is the tuple $\mathcal{E} = \tuple{\event{E}, \stati, \mapsto, \instate}$ such that
 \begin{itemize}
   \item 
   $\stati \subseteq \Pow{\event{E}}$, and
   \item 
   $\mathord{\mapsto} \subseteq \stati\times\stati$ is such that $s\mapsto s'$ implies that
         $s \subset s'$.
 \end{itemize}
 $\instate\in \stati$ is the initial state.
\end{definition}
An event automaton is just a set of subsets of events and a reachability relation $\mapsto$ with the
minimal requirement that if two states $s, s'$ are related by the $\mapsto$ relation, namely
$s\mapsto s'$, then $s'$ is \emph{reached} by $s$ adding at least one event.

\begin{definition}\label{de:ea-simple}
 Let $\mathcal{E} = \tuple{\event{E}, \stati, \mapsto, \instate}$ be an \ea. We say
 that $\mathcal{E}$ is \emph{simple} if 
      $\forall \pmv{e}\in\event{E}$. $\exists s\in \stati$ such that 
      $s\cup\setenum{\pmv{e}}\in \stati$, 
      $s\mapsto s\cup\setenum{\pmv{e}}$ and $s\in \Powfin{\event{E}}$.
\end{definition}
In a simple event automaton, for each event, there is a finite state 
such that this can be reached by adding just this event.

\begin{definition}\label{de:ea-reach}
 Let $\mathcal{E} = \tuple{\event{E}, \stati, \mapsto, \instate}$ be an \ea, and 
 $s \in \stati$ be a state. With 
 $\ragg{s} = \setcomp{s'\in\stati}{s\mapsto s'}$ 
 we denote the subsets of states are 
 reached from $s$ in $\mathcal{E}$.
\end{definition}

\begin{definition}\label{de:ea-compl}
 Let $\mathcal{E} = \tuple{\event{E}, \stati, \mapsto, \instate}$ be an \ea. We
 say that the event automaton $\mathcal{E}$ is \emph{complete} iff 
 each state $s\in \stati$ can be reached by $\instate$.
\end{definition}
A complete event automaton is an event automaton where each states can be reached by
the initial state. Observe that
$\ragg{\cdot}$ of Definition~\ref{de:ea-reach} can be extended to an operator on subsets of states, and 
the operator defined in this way is clearly a monotone and continuous one. 
We can therefore calculate the \emph{least fixed point} of the operator $\ragg{\cdot}$, 
which will be denoted with $\mathit{lfp}(\ragg{\cdot})$, and 
completeness reduces to require that 
$\mathit{lfp}(\ragg{\setenum{\instate}}) = \stati$.

Event automata can easily express configurations of any kind of event structure, 
provided that for each kind a way to reach a configuration from another is given.
The kind of event structure is ranged over by 
$\mu, \mu'\in\setenum{\pes, \fes, r\pes, \dces, \ies, \rces, \Ges}$.

We first state the following theorem, which proof is almost straightforward.
\begin{theorem}\label{th:es-to-ea}
 Let $\mathit{X}$ be an event structure of kind $\mu$ over the set of events $\event{E}$. 
 Then $\grafo{\mu}{\mathit{X}} = \tuple{\event{E}, \Conf{\mathit{X}}{\mu}, \mapsto_{\mu}, \emptyset}$ is 
 an event automaton. 
\end{theorem}
\begin{proof}
 First of all, we observe that for each event structure $\mathit{X}$ of kind $\mu$, with
 $\mu\in\{\pes, \fes,$ $r\pes, \dces, \ies, \rces, \Ges\}$, we have that the empty set belongs
 to $\Conf{\mathit{X}}{\mu}$. 
 Then, due to Definitions~\ref{de:mia-set-configuration}, \ref{de:mapsto-pes}, 
 \ref{de:mapsto-fes}, \ref{de:mapsto-rpes}, \ref{de:mapsto-dces}, \ref{de:mapsto-ies} and 
 \ref{de:mapsto-rces} we have that $s\mapsto_{\mu} s'$ implies that $s \subset s'$ as required
 by the Definition~\ref{de:event-automata}. Hence the thesis follows. 
\end{proof}
The event automata obtained by the configurations of each kind of event structure are simple and
complete. 

\begin{proposition}\label{pr:ea-conf-simple-complete}
 Let $\mathit{X}$ be an event structure of kind 
 $\mu$ over the set of events $\event{E}$ and  
 $\grafo{\mu}{\mathit{X}} = \tuple{\event{E}, \Conf{\mathit{X}}{\mu}, \mapsto_{\mu}, \emptyset}$ the
 associated event automaton. Then $\grafo{\mu}{\mathit{X}}$ is complete and simple.
\end{proposition}
\begin{proof}
 Simplicity is implied by how the $\mapsto_{\mu}$ is defined for the various kind of event structures 
 and by the fact that in each of the considered kind of event structures an event is added at a 
 finite configuration.
 Completeness is a consequence that each configuration of the considered event structures can
 be reached by the initial configuration. 
\end{proof}
Using event automata we can decide when two event structures are equivalent.

\begin{definition}\label{de:es-equivalence}
  Let $\mathit{X}$ and $\mathit{Y}$ be event structures over the same set of events $\event{E}$ 
  of kind $\mu$ and $\mu'$ respectively. 
  We say that  $\mathit{X}$ and $\mathit{Y}$ are \emph{equivalent}, denoted with
  $\mathit{X} \equiv \mathit{Y}$, iff $\grafo{\mu}{\mathit{X}} = \grafo{\mu'}{\mathit{Y}}$.
\end{definition}

The relative expressivity among event structures 
is explicitly studied in \cite{AKPN:DC} and \cite{AKPN:lmcs18}. 
Informally a kind of event structure is more expressive with respect to another, when there is 
a configuration of the former that cannot be a configuration of the latter, whatever is done
with the various relations among events. 
Incomparable means that neither one is more expressive than the other or the vice versa.
We shortly summarize part of these findings, when considering finite configurations.
\pes\ and r\pes\ are equally expressive, whereas \ses\ and \ges\ are strictly more expressive than r\pes, 
and are 
incomparable one with respect to the other.
These two are both less expressive than \dces\ and \rces, which are incomparable.
The relative expressivity of other kinds of event structure has not been
investigated in that paper. 

Here we show that a $\ges$ can be seen as an \ies, adding a tiny piece of information to the
expressivity spectrum of \cite{AKPN:lmcs18}.

\begin{proposition}\label{pr:ges-into-ies}
 Let $\mathit{D} = (\event{E}, \#, \to, \growt)$ be a \ges.
 Let $\mathcal{I}(\mathit{D}) = (\event{E}, \inh)$ be an \ies\ where the relation
 $\inh$ is defined as follows:
 \begin{itemize}
  \item $\forall \pmv{e}, \pmv{e}'\in\event{E}$ such that $\pmv{e}\sharpbin\pmv{e}'$ we have
        $\De{\setenum{\pmv{e}}}{\pmv{e}'}{\emptyset}$ and $\De{\setenum{\pmv{e}'}}{\pmv{e}}{\emptyset}$,
  \item $\forall \pmv{e}, \pmv{e}'\in\event{E}$ such that $\pmv{e} \to \pmv{e}'$ we have
        $\De{\emptyset}{\pmv{e}'}{\setenum{\pmv{e}}}$, and
  \item $\forall \pmv{e}, \pmv{e}', \pmv{e}''\in\event{E}$ such that 
        $\grow{\pmv{e}}{\pmv{e}'}{\pmv{e}''}$ we have
        $\De{\setenum{\pmv{e}}}{\pmv{e}''}{\setenum{\pmv{e}'}}$.              
 \end{itemize} 
 Then $\mathcal{I}(\mathit{D})$ in an \ies\ and $\mathcal{I}(\mathit{D})\equiv\mathit{D}$.
\end{proposition}
\begin{proof}
 Recall that a \ges\ is a \dces\ where the shrinking causality relation is empty. The enabling
 and all the related notions are those of a \dces\ without the part concerning the shrinking causality
 relation.
 
 The fact that $\mathcal{I}(\mathit{D})$ is indeed an \ies\ 
 derives from the fact that the growing causality
 is a kind a weak causality, which is captured by an \ies, thus the $\inh$ relation 
 of $\mathcal{I}(\mathit{D})$ obeys to
 the requirements posed, as each $A$ is either the empty set or a singleton, hence $\# A$.
 
 We now prove that the sets of configurations coincide.
 Take $C\in \Conf{\mathit{D}}{\ges}$, then there is a sequence of events
 $\rho = \pmv{e}_1\pmv{e}_2\cdots\pmv{e}_n\dots$ such that $\toset{\rho} = C$,
 $\CF{C}$ and for each $i\geq 1$ $\toset{\rho_{i-1}}\trans{\pmv{e}_i}$.
 We first show that $C$ is conflict free in the \ies\ interpretation.
 Consider $\pmv{e}\in C$, then $\pmv{e} = \pmv{e}_j$ for some $j\leq\len{\rho}$,
 and take any $\pmv{e}'$ in conflict with $\pmv{e}$ in $\mathit{D}$, then we
 have $\De{\setenum{\pmv{e}}}{\pmv{e}'}{\emptyset}$ and $\De{\setenum{\pmv{e}'}}{\pmv{e}}{\emptyset}$
 and then if $\pmv{e}\in C$ it is impossible that $\pmv{e}'$ is in $C$ as well.
 $\toset{\rho_{i-1}}\trans{\pmv{e}_i}$ in the case $\shrinkt = \emptyset$ reduces 
 to prove that for each $\grow{\pmv{e}}{\pmv{e}'}{\pmv{e}_i}$  if
 $\pmv{e}\in\toset{\rho_{i-1}}$ then also $\pmv{e}'\in\toset{\rho_{i-1}}$, but
 this is exactly what prescribed by $\De{\setenum{\pmv{e}}}{\pmv{e}_i}{\setenum{\pmv{e}'}}$.
 This means that $C\in  \Conf{\mathcal{I}(\mathit{D})}{\ies}$.
 
 Analogously take $C\in  \Conf{\mathcal{I}(\mathit{D})}{\ies}$. 
 It is conflict free
 and it is easy to see that, given the sequence 
 $\rho = \pmv{e}_1\pmv{e}_2\cdots\pmv{e}_n\dots$ such that $\toset{\rho} = C$, for all $i$ we
 have that $\toset{\rho_{i-1}}\trans{\pmv{e}_i}$ as 
 $\ic(\pmv{e}_i) \cup \ac(\overline{\rho_{i-1}}, \pmv{e}_i) \subseteq \overline{\rho_{i-1}}$.
 
 Finally consider $\mapsto_{\ies}$ and $\mapsto_{\ges}$ as defined in Definitions~\ref{de:mapsto-ies}
 and~\ref{de:mapsto-dces}. They clearly coincides, as they are defined starting from the sets of 
 configurations. But this implies that $\mathcal{I}(\mathit{D})\equiv\mathit{D}$.
\end{proof}
We observe that \ies\ are more expressive than \ges, as it is impossible to find a \ges\ such that the
configurations are the one depicted in Figure~\ref{fig:ies-not-ges}.
In fact a \ges\ cannot model this kind of or-causality, as shown also in \cite{AKPN:lmcs18}, 
whereas an \ies\ can.

\begin{figure}[ht]
\centerline{\scalebox{1.1}{\begin{tikzpicture}
[bend angle=45, scale=0.8, 
  read/.style={-,shorten
    <=0pt,thick},
  pre/.style={<-,shorten
    <=0pt,>=stealth,>={Latex[width=1mm,length=1mm]},thick}, post/.style={->,shorten
    >=0,>=stealth,>={Latex[width=1mm,length=1mm]},thick},place/.style={circle, draw=black,
    thick,minimum size=4mm}, transition/.style={rectangle, draw=black!0,
    thick,minimum size=4mm}, invplace/.style={circle,
    draw=black!0,thick}]
\node[transition] (spe) at (0,1.5) {};  
\node[transition] (in) at (0,0) {$\emptyset$};
\node[transition] (a) at (1.5,1) {$\setenum{\pmv{a}}$}
edge[pre] (in);
\node[transition] (c) at (1.5,-1) {$\setenum{\pmv{b}}$}
edge[pre] (in);
\node[transition] (ab) at (3.5,1) {$\setenum{\pmv{a},\pmv{c}}$}
edge[pre] (a);
\node[transition] (ac) at (3.5,-1) {$\setenum{\pmv{b},\pmv{c}}$}
edge[pre] (c);
\end{tikzpicture}}}
\caption{The configuration of the \ies\ with three events $\setenum{\pmv{a},\pmv{b},\pmv{c}}$ and
where the $\inh$ relation is 
$\De{\emptyset}{\pmv{c}}{\setenum{\pmv{a},\pmv{b}}}$, $\De{\setenum{\pmv{a}}}{\pmv{b}}{\emptyset}$
and $\De{\setenum{\pmv{b}}}{\pmv{a}}{\emptyset}$}\label{fig:ies-not-ges}
\end{figure}
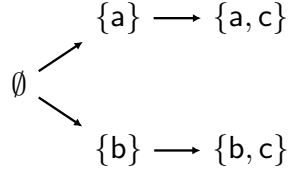

\subsection{Embedding event structures into \Ges:}\ 
We prove now a more general result, namely that given any \emph{event automaton} $\mathcal{E}$, 
which is obtained by the configurations of any kind of event structure, it is possible to obtain a \Ges\ 
whose configurations are precisely the ones of the event automaton $\mathcal{E}$.
We start identifying, in an \ea, the events that are in \emph{conflict}.
The conflict relation we obtain is a \emph{semantic} conflict relation: two events
are in conflict iff they never appear together in a state.

\begin{definition}\label{de:ea-confl}
 Let $\mathcal{E} = \tuple{\event{E}, \stati, \mapsto, \instate}$ be an \ea. 
 We define a symmetric and irreflexive conflict relation $\#_{\ea}$ as follows:
 $\pmv{e}\ \#_{\ea}\ \pmv{e}'$ iff for each $s\in \stati$.\  
 $\setenum{\pmv{e}, \pmv{e}'}\not\subseteq s$.
\end{definition}

In order to obtain the \cd-relation we need some further definitions. 
Fixed an event $\pmv{e}$, the first one identifies the states where this event can be added, and 
the second one identifies the states where the event cannot be added.
\begin{definition}\label{de:ea-caus-pos}
 Let $\mathcal{E} = \tuple{\event{E}, \stati, \mapsto, \instate}$ be an \ea. 
 To each event $\pmv{e}\in \event{E}$ we associate the 
 set of states $\setcomp{s\in \stati}{s\cup\setenum{\pmv{e}}\in \stati\ \land\ s\in\Powfin{\event{E}}\ 
 \land\ s\mapsto s\cup\setenum{\pmv{e}}}$, which we denote with $\Ctr(\mathcal{E},\pmv{e})$.
\end{definition}

\begin{definition}\label{de:ea-caus-neg}
 Let $\mathcal{E} = \tuple{\event{E}, \stati, \mapsto, \instate}$ be an \ea. 
 To each event $\pmv{e}\in \event{E}$ we associate the 
 set of states $\setcomp{s\in \stati}{s\cup\setenum{\pmv{e}}\not\in \stati\ 
 \land\ s\in\Powfin{\event{E}}}$, which 
 we denote with $\Rtc(\mathcal{E},\pmv{e})$.
\end{definition}
Definition~\ref{de:ea-caus-pos} characterizes when an event is enabled giving the \emph{allowing} context,
whereas the Definition~\ref{de:ea-caus-neg} gives the context where the event cannot be added,
and it is called \emph{negative} context.
These two sets are used to obtain the \cd-relation.

\begin{theorem}\label{th:ea-to-Ges}
 Let $\mathcal{E} = \tuple{\event{E}, \stati, \mapsto, \instate}$ be a simple and complete \ea\ 
 such that $\event{E} = \bigcup_{s\in\stati} s$. 
 Then $\toges{\ea}{\mathcal{E}} = 
 (\event{E}, \#,$ $\gesrel)$ is a \Ges, where $\#$ is the relation $\#_{\ea}$ of Definition~\ref{de:ea-confl},
 and for each $\pmv{e}\in \event{E}$ we have 
 $\setcomp{(X,\emptyset)}{X\in \Ctr(\mathcal{E},\pmv{e})}\cup
 \setcomp{(X,\setenum{\pmv{e}})}{X\in \Rtc(\mathcal{E},\pmv{e})}\gesrel \pmv{e}$. 
 Furthermore $\mathcal{E} \equiv \grafo{\Ges}{\toges{\ea}{\mathcal{E}}}$.
\end{theorem}
\begin{proof}
 $\toges{\ea}{\mathcal{E}}$ is clearly a \Ges: the conflict relation obtained 
 using the Definition~\ref{de:ea-confl} is trivially symmetric and it is 
 clearly irreflexive
 as each event appear in a state and the event automaton is simple and complete.
 For what concern the $\gesrel$ relation, given $\pmv{Z}\gesrel \pmv{e}$, clearly
 $\pmv{Z}\neq \emptyset$, furthermore for each $(X,Y)\in \pmv{Z}$ we have that
 $\CF{X}$, as $X$ is a finite state of the \ea, and also $\CF{Y}$ holds, as $Y$
 is either the empty set or e singleton, and finally, given two elements
 $(X,Y)$ and $(X',Y')$ of the entry $\pmv{Z}\gesrel \pmv{e}$, we have that
 always $X \neq X'$, and then also the last condition for the $\gesrel$ holds.
 
 We show now that, given a state $s\in \stati$, $s$ is also a configuration of
 the \Ges\  $\toges{\ea}{\mathcal{E}}$.
 As $\mathcal{E}$ is simple and complete, there exists a sequence of events
 $\rho = \pmv{e}_1\cdots\pmv{e}_n\cdots$ such that $\forall i \geq 1$
 $\toset{\rho_i}\in \stati$ and $\toset{\rho_i} \mapsto \toset{\rho_i}\cup\setenum{\pmv{e}_{i+1}}$.
 For each $i$, we have that $\ctx(\pmv{Z}\gesrel\pmv{e}_{i+1})\cap\toset{\rho}$ 
 is exactly $\toset{\rho_i}$, and the element
 $(\toset{\rho_i},\emptyset)$ belongs to the entry
 $\pmv{Z}\gesrel\pmv{e}_{i+1}$, hence $\toset{\rho_i}\trans{\pmv{e}_{i+1}}$ as required.
 Clearly the $\mapsto_{\toges{\ea}{\mathcal{E}}}$ coincides with $\mapsto$. 
 Finally we observe that the \Ges\ $\toges{\ea}{\mathcal{E}}$ is full and faithful. 
 
 The vice versa follows the same line. Consider a configuration 
 $C\in \Conf{\toges{\ea}{\mathcal{E}}}{\Ges}$, then there exists a sequence 
 $\rho = \pmv{e}_1\cdots\pmv{e}_n\cdots$ of events such that $\forall i \geq 1$
 it holds that  $\toset{\rho_i}\trans{\pmv{e}_{i+1}}$. We prove something
 stronger, namely that each $\toset{\rho_i}$ is a state, and we do by induction
 on the length of the sequence. The basis is trivial, assume then it hold for $n$, and
 we have that  $\toset{\rho_n}\trans{\pmv{e}_{n+1}}$. But the \ea\ $\mathcal{E}$ is simple and
 complete, and by the way the $\pmv{Z}\gesrel\pmv{e}_{n+1}$ is constructed,
 we have that $\toset{\rho_n}\cup\setenum{\pmv{e}_{n+1}}$ is a state and the thesis follows.
 Again $\mapsto_{\toges{\ea}{\mathcal{E}}}$ coincides with $\mapsto$, thus 
 $\mathcal{E} \equiv \grafo{\Ges}{\toges{\ea}{\mathcal{E}}}$.
\end{proof}

The theorem has a main consequence, namely that event automata and \Ges\ are equally expressive. 

\begin{example}
 Consider the \rces\ of the Example~\ref{ex:rces}. The associated event automaton is the one
 depicted in the Example~\ref{ex:new2res}.
 It has no conflict as all the three events are present in a 
 configuration together.
 The associated  \cd-relation, 
 obtained using Definition~\ref{de:ea-caus-pos} and Definition~\ref{de:ea-caus-neg}, is the following one, 
 which is a little different from the one devised in the Example~\ref{ex:new2res} as here it is obtained 
 from an event automaton.
 We synthesize 
 $$\{(\emptyset,\emptyset),(\setenum{\pmv{c}},\emptyset),(\setenum{\pmv{c},\pmv{b}},\emptyset),
 (\setenum{\pmv{b}},\setenum{\pmv{a}})\}\gesrel \pmv{a}$$
 because the set $\Ctr(\Conf{\mathit{R}}{\rces},\pmv{a})$ contains the sets $\emptyset$, $\setenum{\pmv{c}}$ and 
 $\setenum{\pmv{c},\pmv{b}}$, whereas the set of the \emph{negative context} $\Rtc(\Conf{\mathit{R}}{\rces},\pmv{a})$ 
 contains just $\setenum{\pmv{b}}$,
$$\{(\emptyset,\emptyset),(\setenum{\pmv{c}},\emptyset),(\setenum{\pmv{a},\pmv{c}},\emptyset),
 (\setenum{\pmv{a}},\setenum{\pmv{b}})\}\gesrel \pmv{b}$$
 as $\Ctr(\Conf{\mathit{R}}{\rces},\pmv{b})$ contains the sets $\emptyset$, $\setenum{\pmv{c}}$ and 
 $\setenum{\pmv{a},\pmv{c}}$, $\Rtc(\Conf{\mathit{R}}{\rces},\pmv{b})$ contains $\setenum{\pmv{a}}$, 
 and finally
 $$\setenum{(\emptyset,\emptyset),(\setenum{\pmv{a}},\emptyset),(\setenum{\pmv{b}},\emptyset)}
 \gesrel \pmv{c}$$
  as $\Ctr(\Conf{\mathit{R}}{\rces},\pmv{c})$ contains the sets $\emptyset$, $\setenum{\pmv{a}}$ and 
 $\setenum{\pmv{b}}$, and $\Rtc(\Conf{\mathit{R}}{\rces},\pmv{c})$ is the empty set.
\end{example} 
As a consequence of the Theorem~\ref{th:ea-to-Ges} we have the following result.
\begin{corollary}\label{co:every-evstruct-is-Ges}
  Let $\mathit{X}$ be an event structure of type $\mu$ and let $\grafo{\mu}{\mathit{X}}$ be the associated
  \ea. Then  $\toges{\ea}{\grafo{\mu}{\mathit{X}}}$ is \Ges, and 
  $\mathit{X} \equiv \toges{\ea}{\grafo{\mu}{\mathit{X}}}$.
\end{corollary}
 The construction identifies properly the \emph{context} in which an event is allowed to happen, and this 
 context
 becomes the main ingredient of the \cd-relation, as the construction does not give the \emph{causes} but
 just the context. 
 If on the one hand this suggests that the context, rather than the causal dependencies, is the relevant
 ingredient,
 on the other hand it is less informative with respect to the usual causality definitions.

 We review some kind of event structures, showing that a more informative \cd-relation can be indeed
 obtained. We will focus only on few of them.

\subsection{Prime event structures:}\ In this case the idea is that causes of an event are just the
set of events that should be present in the configuration. 

\begin{proposition}\label{pr:pes-to-Ges}
 Let $\mathit{P} = (\event{E}, \leq, \#)$ be a \pes. Then $\toges{\pes}{\mathit{P}} = 
 (\event{E}, \#, \gesrel)$ is a \Ges, where 
 $\setenum{(\emptyset,\hist{\pmv{e}}\setminus\setenum{\pmv{e}})}\gesrel \pmv{e}$ for 
 each $\pmv{e}\in\event{E}$. Furthermore $\mathit{P} \equiv \toges{\pes}{\mathit{P}}$.
\end{proposition} 
\begin{proof}
 $\toges{\pes}{\mathit{P}}$ is trivially a \Ges. We show that 
 $\Conf{\toges{\pes}{\mathit{P}}}{\Ges} = \Conf{P}{\pes}$.
 First we observe that, for each event $\pmv{e}$, we have that $\hist{\pmv{e}}\in \Conf{P}{\pes}$. 
 Clearly to $\hist{\pmv{e}}$ we can associate a sequence $\rho = \pmv{e}_1\cdots \pmv{e}_n$ with
 $\pmv{e}_n = \pmv{e}$ and for each $1 \leq i\leq n$ we have that 
 $\toset{\rho_i} \in \Conf{P}{\pes}$. This sequence can be used to show that also 
 each $\toset{\rho_i} \in \Conf{\toges{\pes}{\mathit{P}}}{\Ges}$. 
 Given a configuration $C\in \Conf{\mathit{P}}{\pes}$ we can associate
 to it a sequence $\rho = \pmv{e}_1\cdots \pmv{e}_n\cdots$ and for each $i$ we have again
 that $\toset{\rho_i} \in \Conf{P}{\pes}$. Now for each $\pmv{e}_{i+i}$ we have that
 $\toset{\rho_i} \cap 
 \ctx(\setenum{(\emptyset,\hist{\pmv{e}_{i+i}}\setminus\setenum{\pmv{e}_{i+i}})}\gesrel \pmv{e}_{i+i}) =
 \emptyset$ and clearly $\hist{\pmv{e}_{i+i}}\setminus\setenum{\pmv{e}_{i+i}}\subseteq \toset{\rho_i}$
 which means that $\pmv{e}_{i+i}$ is enabled at $\toset{\rho_i}$, hence 
 $C\in \Conf{\toges{\pes}{\mathit{P}}}{\Ges}$. This proves that each configuration of a \pes\ 
 $\mathit{P}$ is also a configuration of the \Ges\ $\toges{\pes}{\mathit{P}}$.
 The vice versa holds as well observing that, given a configuration
 $C\in \Conf{\toges{\pes}{\mathit{P}}}{\Ges}$ and the associated 
 $\rho = \pmv{e}_1\cdots \pmv{e}_n\cdots$, the fact that $\toset{\rho_i}\trans{\pmv{e}_{i+i}}$
 means that $\hist{\pmv{e}_{i+i}}\subseteq \toset{\rho_{i+1}} \subseteq C$.
 The $\mapsto_{\pes}$ and $\mapsto_{\Ges}$ associated at these event structures clearly coincide. 
 Hence $\mathit{P} \equiv \toges{\pes}{\mathit{P}}$.
\end{proof}
The \cd-relation defined in the previous proposition is a bit more informative with respect to
the one defined in Theorem~\ref{th:ea-to-Ges}. In fact it captures the intuition that in a \pes\
there are not modifiers. 
We stress that is not the unique way to associate to the causality relation $\leq$ of a \pes\  the 
$\gesrel$ relation: one alternative would have been to add 
$\setenum{(\emptyset,\setenum{\pmv{e}'})}\gesrel \pmv{e}$
for each $\pmv{e}' < \pmv{e}$ and another one would be 
$\setenum{(\hist{\pmv{e}}\setminus\setenum{\pmv{e}},\emptyset)}\gesrel \pmv{e}$ showing that the
events causally before $\pmv{e}$ are indeed the context allowing the event $\pmv{e}$ to happen.

\begin{example}
 Consider the \pes\ $(\setenum{\pmv{a}, \pmv{b}, \pmv{c}}, \leq, \#)$ where $\pmv{a}\leq \pmv{b}$ (we omit 
 the reflexive part of the $\leq$ relation),  
 $\pmv{a}\sharpbin\pmv{c}$ and $\pmv{b}\sharpbin\pmv{c}$. The event traces are 
 $\epsilon$, $\pmv{a}$, $\pmv{a}\pmv{b}$
 and $\pmv{c}$, and the associated configurations are $\emptyset$, 
 $\setenum{\pmv{a}}$, $\setenum{\pmv{a}, \pmv{b}}$
 and $\setenum{\pmv{c}}$ (the $\mapsto_{\pes}$ relation is obvious).
 The conflict relation is the same and the
 \cd-relation is $\setenum{(\emptyset,\emptyset)}\gesrel \pmv{a}$, 
 $\setenum{(\emptyset,\emptyset)}\gesrel \pmv{c}$
 and $\setenum{(\emptyset,\setenum{\pmv{a}})}\gesrel \pmv{b}$. As noticed before we could have 
 stipulated also
 $\setenum{(\setenum{\pmv{a}},\emptyset)}\gesrel \pmv{b}$ instead of $\setenum{(\emptyset,\setenum{\pmv{a}})}\gesrel \pmv{b}$ obtaining the same set of configurations and the same transition graph.
\end{example}

\subsection{Relaxed prime event structures:}\ Also in this case the idea is similar to
\pes, as the set of immediate causes of an event $\pmv{e}$ 
can be considered akin to the local configuration of $\pmv{e}$.

\begin{proposition}\label{pr:rpes-to-Ges}
 Let $\mathit{T} = (\event{E}, \#, \to)$ be a r\pes. Then $\toges{r\pes}{\mathit{T}} = 
 (\event{E}, \#, \gesrel)$ is a \Ges, where 
 $\setenum{(\emptyset,\ic(\pmv{e}))}\gesrel \pmv{e}$ for 
 each $\pmv{e}\in\event{E}$. Furthermore $\mathit{T} \equiv \toges{r\pes}{\mathit{T}}$.
\end{proposition} 
\begin{proof}
 Like the one of Proposition~\ref{pr:pes-to-Ges}, observing that 
 $\ic(\pmv{e})$ plays the same role as $\hist{\pmv{e}}$.
\end{proof}
Observe that in case an event $\pmv{e}$ has no causes, the $\ic(\pmv{e})$ set is empty, as expected.
Again this is not the unique way to associate to the causality relation $\to$ of a r\pes\  the 
$\gesrel$ relation. We could have defined $\setenum{(\emptyset,\setenum{\pmv{e}'})}\gesrel \pmv{e}$
for each $\pmv{e}' \to \pmv{e}$, or  
$\setenum{(\ic(\pmv{e}),\emptyset)}\gesrel \pmv{e}$.

\subsection{Flow event structure:}\
To give a direct translation of a \fes\ into a \Ges\ we need some auxiliary notation. Let 
$\mathit{F} = (\event{E},\prec, \#)$, with 
$\fl{\pmv{e}} = \setcomp{\pmv{e}'\in \event{E}}{\pmv{e}' \prec \pmv{e}}$ we denote the set of events 
\emph{preceding} $\pmv{e}$. Observe that this set may contains conflicting events.
Starting form the set $\fl{\pmv{e}}$ we identifies the subsets of events that are
conflict free and maximal. Thus we have the
set $\maxfl{\pmv{e}} = \setcomp{X\subseteq \fl{\pmv{e}}}{\CF{X}\ \land\ 
\neg \CF{X\cup\setenum{\pmv{e}'}}\ \mathit{with}\ \pmv{e}'\in \fl{\pmv{e}}\setminus X}$
which contains all subsets $X$ of $\fl{\pmv{e}}$ which are conflict free ($\CF{X}$) and
maximal ($\neg \CF{X\cup\setenum{\pmv{e}'}}$ where $\pmv{e}'\in \fl{\pmv{e}}\setminus X$).
These subsets will play the same role played by the immediate causes for an event.

\begin{proposition}\label{pr:fes-to-Ges}
 Let $\mathit{F} = (\event{E}, \prec, \#)$ be a full and faithful \fes. 
 Then $\toges{\fes}{\mathit{F}} = 
 (\event{E}, \#, \gesrel)$ is a \Ges, where for each $\pmv{e}\in\event{E}$ we have
 $\setcomp{(X,\emptyset)}{X\in\maxfl{\pmv{e}}} \gesrel \pmv{e}$. Furthermore $\mathit{F} \equiv \toges{\fes}{\mathit{F}}$.
\end{proposition} 
\begin{proof}
 Clearly $\toges{\fes}{\mathit{F}}$ is a \Ges. The 
 conflict relation is irreflexive, and this is a consequence that $\mathit{F}$ is full and
 faithful, and each event $\pmv{e}$ is such that there is an
 entry $\pmv{Z}\gesrel \pmv{e}$, as in the case $\fl{\pmv{e}} = \emptyset$ we have
 that $\maxfl{\pmv{e}}$ contains just the emptyset.
 
 Now we prove that $\Conf{\toges{\fes}{\mathit{F}}}{\Ges} = \Conf{F}{\fes}$.
 Consider $C\in \Conf{F}{\fes}$. It is conflict free and the transitive and
 reflexive closure of $\prec$ is a partial ordering on $C$. Consider then  
 the sequence $\rho = \pmv{e}_1\pmv{e}_2\cdots\pmv{e}_n\cdots$ compatible with
 the $\prec^{*}$, we have to show that for each $i\geq 1$ we have that
 $\toset{\rho_i}\trans{\pmv{e}_{i+1}}$. Consider $\pmv{e}\prec \pmv{e}_{i+1}$ and
 assume that $\pmv{e}\not\in \toset{\rho_i}$, then there must be
 a $\pmv{e}'\sharpbin\pmv{e}$ such that $\pmv{e}'\prec \pmv{e}_{i+1}$ and 
 $\pmv{e}'\in \toset{\rho_i}$, and this for all the events in $\fl{\pmv{e}_{i+1}}$,
 which means that there must be a maximal subset of non conflicting events that
 are in the flow relation with $\pmv{e}_{i+1}$, but this implies that
 $\ctx(\pmv{Z}\gesrel \pmv{e}_{i+1})\cap \toset{\rho_i} = X \in \maxfl{\pmv{e}_{i+1}}$, and
 we can conclude that $C$ is a configuration of $\toges{\fes}{\mathit{F}}$.
 
 Consider $C\in \Conf{\toges{\fes}{\mathit{F}}}{\Ges}$. Then there exists a sequence 
 $\rho = \pmv{e}_1\pmv{e}_2\cdots\pmv{e}_n\cdots$ such that for each $i\geq 1$, 
 $\toset{\rho_{i-1}}\trans{\pmv{e}_{i}}$. This means that
 $\ctx(\pmv{Z}\gesrel \pmv{e}_{i})\cap \toset{\rho_{i-1}} = X \in \maxfl{\pmv{e}_{i}}$.
 Now consider $\pmv{e}\prec \pmv{e}_{i}$ and assume that $\pmv{e}\not\in\toset{\rho_{i-1}}$,
 as $X\in \maxfl{\pmv{e}_{i}}$ by maximality of $X$ we have that there is an
 event $\pmv{e}_j\in\toset{\rho_{i-1}}$ with $j<i$ such that $\pmv{e}_j\sharpbin\pmv{e}$.
 The fact that the reflexive and transitive closure of $\prec$ is a partial order 
 on $C$ is trivial. We can conclude that $C\in \Conf{F}{\fes}$.
 
 The $\mapsto$ relations induced by the two event structure trivially coincide,
 and the thesis $\mathit{F} \equiv \toges{\fes}{\mathit{F}}$ follows.
\end{proof}

\begin{example}\label{ex:fes-to-Ges}
 Consider the flow event structure $\mathit{F}$ with $4$ events $\pmv{e}_1, \pmv{e}_2, \pmv{e}_3$ and
 $\pmv{e}$ and where $\pmv{e}_i \prec \pmv{e}$, for $1\leq i\leq 3$, and
 $\pmv{e}_1\sharpbin\pmv{e}_2$. We have that $\maxfl{\pmv{e}_i} = \setenum{\emptyset}$ for
 $1\leq i\leq 3$ and
 $\maxfl{\pmv{e}} = \setenum{\setenum{\pmv{e}_1, \pmv{e}_3}, \setenum{\pmv{e}_2, \pmv{e}_3}}$
 and then we have $\setenum{(\emptyset,\emptyset)}\gesrel \pmv{e}_i$ for $1\leq i\leq 3$ and
 $\setenum{(\setenum{\pmv{e}_1, \pmv{e}_3},\emptyset),(\setenum{\pmv{e}_2, \pmv{e}_3},\emptyset)}
 \gesrel \pmv{e}$, and the conflict relation is $\pmv{e}_1\sharpbin\pmv{e}_2$.
 The configurations of this \fes\ (and of the \Ges\ $\toges{\fes}{\mathit{F}}$) are depicted in
 Figure~\ref{fig:fes-to-Ges}.
 
 \begin{figure}[ht]
 \centerline{\scalebox{1.1}{\begin{tikzpicture}
[bend angle=45, scale=0.9, 
  read/.style={-,shorten
    <=0pt,thick},
  pre/.style={<-,shorten
    <=0pt,>=stealth,>={Latex[width=1mm,length=1mm]},thick}, post/.style={->,shorten
    >=0,>=stealth,>={Latex[width=1mm,length=1mm]},thick},place/.style={circle, draw=black,
    thick,minimum size=4mm}, transition/.style={rectangle, draw=black!0,
    thick,minimum size=4mm}, invplace/.style={circle,
    draw=black!0,thick}]
\node[transition] (spe) at (0,1.5) {};  
\node[transition] (in) at (0,0) {$\emptyset$};
\node[transition] (a) at (1.8,1) {$\setenum{\pmv{e}_2}$}
edge[pre] (in);
\node[transition] (b) at (1.8,-1) {$\setenum{\pmv{e_1}}$}
edge[pre] (in);
\node[transition] (c) at (1.8,0) {$\setenum{\pmv{e_3}}$}
edge[pre] (in);
\node[transition] (bc) at (4,-1) {$\setenum{\pmv{e}_1,\pmv{e}_3}$}
edge[pre] (c)
edge[pre] (b);
\node[transition] (ac) at (4,1) {$\setenum{\pmv{e}_2,\pmv{e}_3}$}
edge[pre] (a)
edge[pre] (c);
\node[transition] (abc) at (7,1) {$\setenum{\pmv{e}_1,\pmv{e}_3,\pmv{e}}$}
edge[pre] (ac);
\node[transition] (abc) at (7,-1) {$\setenum{\pmv{e}_2,\pmv{e}_3,\pmv{e}}$}
edge[pre] (bc);
\end{tikzpicture}}}
 \caption{The configurations of the \fes\ and \Ges\ in the 
 Example~\ref{ex:fes-to-Ges}}\label{fig:fes-to-Ges}
 \end{figure}
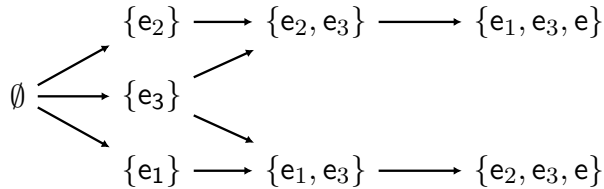
\end{example}

\subsection{Dynamic causality event structures:}\ 
The intuition in this case is to code all the possible subsets of modifiers for a given target,
and for each subset of modifiers, determine what is the set of enabling events. In this way 
the $\gesrel$ relation can be easily obtained.
\begin{proposition}\label{pr:dces-to-Ges}
 Let $\mathit{D} = (\event{E}, \#, \to, \shrinkt, \growt)$ be a \dces.\  $\toges{\dces}{\mathit{D}} = 
 (\event{E}, \#,$ ${\gesrel)}$ is a \Ges\ where the relation $\gesrel$ is defined as 
 $\setcomp{\big(X,\big(\ic(\pmv{e}) \setminus\bigcup_{\pmv{e}'\in X}\drop{\pmv{e}'}{\pmv{e}}\big)\ \cup\  
 \bigcup_{\pmv{e}'\in X}\agg{\pmv{e}'}{\pmv{e}}\big)}{X\subseteq \groset{\pmv{e}}\cup\shrset{\pmv{e}}}
 \gesrel \pmv{e}$ for each $\pmv{e}\in\event{E}$.
 Furthermore $\mathit{D}\equiv \toges{\dces}{\mathit{D}}$.
\end{proposition}
\begin{proof}
 $\toges{\dces}{\mathit{D}}$ is clearly a \Ges.
 We show that the sets of configurations coincide,
 which is enough to show that $\mathit{D}\equiv \toges{\dces}{\mathit{D}}$.
 Take $C\in \Conf{\mathit{D}}{\dces}$, then there is sequence of events
 $\rho = \pmv{e}_1\pmv{e}_2\cdots\pmv{e}_n\cdots$ and it is such that
 for each $i \geq 1$ we have that
 $((\ic(\pmv{e}_i) \cup \ac(\toset{\rho_{i-1}}, \pmv{e}_i)) \setminus 
 \dc(\toset{\rho_{i-1}}, \pmv{e}_i)) \subseteq \toset{\rho_{i-1}}$, now
 the entry associated to $\pmv{e}_i$ is $\pmv{Z}_i\gesrel \pmv{e}_i$ and
 each element $(X,Y)$ of this entry is such that $X$ is subset of modifiers.
 Consider the $X$ such that 
 $\ctx(\pmv{Z}_i\gesrel \pmv{e}_i) \cap \toset{\rho_{i-1}} = X$, then
 $Y = ((\ic(\pmv{e}_i) \cup \ac(\toset{\rho_{i-1}}, \pmv{e}_i)) \setminus 
 \dc(\toset{\rho_{i-1}}, \pmv{e}_i)) \subseteq \toset{\rho_{i-1}}$ which is
 exactly what is required, hence $C\in \Conf{\toges{\dces}{\mathit{D}}}{\Ges}$
 as well.
 
 The vice versa, namely that each configuration $C\in \Conf{\toges{\dces}{\mathit{D}}}{\Ges}$
 is a configuration of $\mathit{D}$, is analogous.
\end{proof}

\begin{example}
 Concerning the \dces\ of the example~\ref{ex:dces}, the conflict relation is the one of the
 \dces\, whereas the \cd-relation is 
 $\setenum{(\emptyset,\emptyset)}\gesrel \pmv{a}$, $\setenum{(\emptyset,\emptyset)}\gesrel \pmv{b}$,
 $\setenum{(\emptyset,\emptyset)}\gesrel \pmv{e}$, $\setenum{(\emptyset,\emptyset)}\gesrel \pmv{d}$ and
 for $\pmv{c}$ we have
 $\setenum{(\emptyset,\setenum{\pmv{b}}),(\setenum{\pmv{a}},\emptyset),
 (\setenum{\pmv{d}},\setenum{\pmv{b},\pmv{e}}),(\setenum{\pmv{a},\pmv{d}},\setenum{\pmv{e}})}\gesrel \pmv{c}$.
\end{example} 

\subsection{Inhibitor event structures:}\ 
In the case of \ies\ there are two main observations to be done: 
one, there is no conflict relation, and second,
though there is some similarity between the $\inh$ relation and the $\gesrel$ relation, 
there is also a quite subtle 
difference. 
When adding an event $\pmv{e}$ to a configuration of an \ies, and we have $\De{a}{\pmv{e}}{A}$, 
one would simply 
add the pairs $(a,\setenum{\pmv{e}'})$ for each $\pmv{e}'\in A$ (as the events in $A$ are 
pairwise conflicting) but this does not work in the case $A$ is the empty set, 
as it has a different meaning in the $\inh$ relation  with respect to the $\gesrel$ relation. 
In the former, it means that the event in $a$ inhibits the event $\pmv{e}$, whereas in the latter
the pair $(a,\emptyset)$ simply says that if the context $a$ is present then there is no 
further event needed. 
Taking into account these differences, the translation is fairly simple.  

\begin{proposition}\label{pr:ies-to-Ges}
 Let $\mathit{I} = (\event{E},\inh)$ be an \ies. $\toges{\ies}{\mathit{I}} = 
 (\event{E}, \#,$ $\gesrel)$ is a \Ges, where 
 \begin{itemize}
   \item 
    $\pmv{e}\sharpbin\pmv{e}'$ iff $\De{\setenum{\pmv{e}}}{\pmv{e}'}{\emptyset}$ and 
         $\De{\setenum{\pmv{e}'}}{\pmv{e}}{\emptyset}$ , and
   \item 
    for each $\pmv{e}\in\event{E}$, if $\De{a}{\pmv{e}}{A}$ and $A\neq\emptyset$ then 
         $\setcomp{(\emptyset,\emptyset)}{a\neq\emptyset}\cup
         \setcomp{(a\cup\setenum{\pmv{e}'},\emptyset)}{\pmv{e}'\in A}\gesrel \pmv{e}$, 
         if $\De{a}{\pmv{e}}{A}$ and $A = \emptyset$ then 
         $\setenum{(a,\setenum{\pmv{e}})}\gesrel \pmv{e}$.
 \end{itemize} 
 Furthermore $\mathit{I} \equiv \toges{\ies}{\mathit{I}}$.
\end{proposition}
\begin{proof}
 The $\gesrel$ relation of $\toges{\ies}{\mathit{I}}$ obeys to the requirements of
 Definition~\ref{de:mia-es} and the conflict relation $\#$ is a symmetric and irreflexive 
 relation. 
 
 We prove that $\Conf{\mathit{I}}{\ies} = \Conf{\toges{\ies}{\mathit{I}}}{\Ges}$.
 Consider $C\in \Conf{\mathit{I}}{\ies}$, then we have that there exists a sequence of events
 $\rho = \pmv{e}_1\pmv{e}_2\cdots\pmv{e}_n\cdots$ such that for each $\pmv{e}_i$ and
 and $\De{a}{\pmv{e}_i}{A}$ we have that 
 $a\subseteq \toset{\rho_{i-1}}\ \Rightarrow\ \toset{\rho_{i-1}}\cap A\neq\emptyset$. 
 Assume that $a\subseteq \toset{\rho_{i-1}}$ holds, and $A\neq\emptyset$,
 then we have the entry 
 $\pmv{Z} \gesrel \pmv{e}_i$, where
 $\pmv{Z} = \setcomp{(\emptyset,\emptyset)}{a\neq\emptyset}\cup
 \setcomp{(a\cup\setenum{\pmv{e}'},\emptyset)}{\pmv{e}'\in A}$, 
 and then also the enabling condition of the \Ges\ is satisfied, for each entry obtained 
 for the event $\pmv{e}_i$, as in the case $a\neq\emptyset$ and 
 $\ctx{\pmv{Z} \gesrel \pmv{e}_i}\cap\toset{\rho_{i-1}} = \emptyset$ we have
 that the event can be added, and in the other case, if $a\neq\emptyset$, then
 one of the event in $A$ must be present in $\toset{\rho_{i-1}}$, which is captured
 by requiring that the context is $a$ together with this event. We can conclude that  
 $C$ is also a configuration in $\Conf{\toges{\ies}{\mathit{I}}}{\Ges}$.
 
 For the vice versa, assume that $C\in \Conf{\toges{\ies}{\mathit{I}}}{\Ges}$, hence
 we have $\rho = \pmv{e}_1\pmv{e}_2\cdots\pmv{e}_n\cdots$ and 
 for each $i\geq 1$ it holds that $\toset{\rho_{i-1}}\trans{\pmv{e}_i}$.
 This means that for each entry $\pmv{Z} \gesrel \pmv{e}_i$, there is
 an element $(X,\emptyset)\in \pmv{Z}$ such that 
 $\ctx(\pmv{Z} \gesrel \pmv{e}_i)\cap $\De{a}{\pmv{e}}{A}$ = X$. But if
 $X \neq \emptyset$ then $X = a\cup\setenum{\pmv{e}}$ for some $\pmv{e}\in A$
 for a $\De{a}{\pmv{e}}{A}$, and this implies that 
 $a\subset \toset{\rho_{i-1}}\ \Rightarrow \toset{\rho_{i-1}}\cap A\neq\emptyset$, hence
 $C$ is also a configuration in $\Conf{\mathit{I}}{\ies}$.
 
 The $\mapsto$ in both event structures coincide, hence $\mathit{I} \equiv \toges{\ies}{\mathit{I}}$.  
\end{proof}
 
\begin{example}
 The \ies\ of the example~\ref{ex:ies} induces the empty conflict relation, and
 the \cd-relation is $\setenum{(\emptyset,\emptyset)}\gesrel \pmv{a}$, 
 $\setenum{(\emptyset,\setenum{\pmv{a}})}\gesrel \pmv{b}$
 and $\setenum{(\emptyset,\emptyset), (\setenum{\pmv{a}},\setenum{\pmv{b}})}\gesrel \pmv{c}$.
\end{example}

\subsection{Higher order causality event structures:}\
The comparison with event structures with higher-order dynamics of \cite{KN15:higher} is done indirectly, 
as these
are equivalent to event structures with resolvable conflicts. 
In this approach the relations $\shrinkt$ and $\growt$ are generalized to take into account set of 
modifiers, targets and contributions. 
The drawback is that the happening of an event implies a recalculation of these relation, 
similarly to what it
is done in causal automata.
In fact it is fairly obvious that
given  one simple step transition graph (meaning that a configuration is reached by another one adding 
just one event), it is always possible to obtain a \Ges.


\section{Conclusion}\label{sec:conc}
In this paper we have introduced a new brand of event structure where the main relation,
the \cd-relation, models the various conditions under which an event can be added to a subset of events.
The relation is now defined as $\gesrel\ \subseteq \Pow{\pmv{A}}\times \event{E}$,  
where $\pmv{A} \subseteq \Powfin{\event{E}}\times\Powfin{\event{E}}$, thus it stipulates for each event 
which are the context-dependency pairs, but it can be easily generalized to subsets of events
modeling precisely, when events happen together (as it is done in \cite{PS:tcs12} or \cite{GP:ESRC}). 
The focus is on the contexts in which an event can be added, which may change, rather that modeling 
the dependencies and how these may change. 
Here the choice is whether it is better to focus on dependencies (and how they may change) or on 
the context.
The advantage of the latter is its generality, whereas the former may be useful in pointing out relations
among events.

It should be clear that this kind of event structures is capable of modeling the same enabling situation
for an event in various way, and it could be interesting to understand if there could be an informative 
way canonically. 
In fact, the canonical relation just focus on all the contexts in which an event can be added, and
the dependency set is less informative. 
Thus finding a way to identify minimal contexts together with a set of dependencies may be useful, 
similarly
to what it has been discussed when associating \pes\ to \Ges.

It remains to stress that \Ges\ can be generalized not only allowing steps but also representing contexts 
in a richer way.
Here we have considered contexts as subset of events, but they can have a richer structure. This would 
allow
to characterize more precisely contexts, allowing, for instance, to drop the last requirement we have 
placed on \dces,
as in this case the order in which the modifiers appear may influence the dependencies.
Finally we observe that the idea of context is not new, for instance they have been considered
in \cite{LeiferM06} or in \cite{BaldanBB07}, and a comparison with these should be considered.

In this paper we have considered various event structures, still some interesting notions remained
out of the scope of this paper, like reversible event structures \cite{iainirek}, but we are confident 
that our approach can be used also in the reversibility setting, clearly by introducing a suitable
relation $\ll$ for the reversing events and upon the identification of the context in which the
reversing event can be \emph{performed}.
We do not have considered event structures with circular dependencies (\cite{BCPZ:CircCausES}) 
basically because in this case the configurations would not give an \ea\ like those used here.
In fact the configurations of an event structure with circular dependencies are pair of subsets 
of events, those actually happened and those that must happen to guarantee that the circular 
dependencies are fulfilled.
It is however interesting to understand how a context in this way can be used, the intuition
being that somehow it should be taken into account those events that have still to be performed as
required by some circular dependency. 

Finally we would like to mention that two interesting research issues regard the categorical 
treatment of this new brand of event structure and their relation to Petri nets. 
The categorical treatment could be inferred from the categorical treatment of event automata and
the enabling/disabling relation studied in \cite{Pi:FI06}, whereas for the second one we
can follow the lines introduced in \cite{CaPi:SAC17}, where Petri nets are related to
dynamic event structures.

\section*{Acknowledgment}
I wish to thank the reviewers and the Coordination 2019 audience for their useful suggestions and
criticisms.


\bibliographystyle{alpha}
\bibliography{main}
\end{document}